\documentclass{ws-us}[draft]
\usepackage{cite}
\usepackage{url}
\usepackage{amsmath}
\usepackage{bm}
\usepackage{color}
\usepackage{setspace}

\newcommand{\red}[1]{\textcolor{black}{#1}}
\newcommand{\blue}[1]{\textcolor{black}{#1}}

\begin{document}
\setlength{\pdfpagewidth}{8.5in}
\setlength{\pdfpageheight}{11in}
\catchline{0}{0}{2013}{}{}

\markboth{Zirui Niu, Giordano Scarciotti, and Alessandro Astolfi}{Model Reduction for Hybrid Systems with State-dependent Jumps}

\title{Interconnection-based Model Reduction for Linear Hybrid Systems}

\author{Zirui Niu$^{a}$, Giordano Scarciotti$^a$, and Alessandro Astolfi$^{a, b}$}

\address{$^a$Department of Electrical and Electronic Engineering, Imperial College London, London SW7 2AZ, UK\\
E-mails: [zirui.niu20,g.scarciotti,a.astolfi]@imperial.ac.uk}

\address{$^b$Dipartimento di Ingegneria Civile e Ingegneria Informatica, Universit{\`a} di Roma ``Tor Vergata'', Via del Politecnico 1, 00133 Roma, Italy}

\maketitle

\begin{abstract}
In this paper, we address the model reduction problem for linear hybrid systems via the interconnection-based technique called moment matching. We consider two classical interconnections, namely the direct and swapped interconnections, in the hybrid setting, and we present families of reduced-order models for each interconnection via a hybrid characterisation of the steady-state responses. By combining the results for each interconnection, the design of a reduced-order model that achieves moment matching simultaneously for both interconnections is studied. In addition, we show that the presented results have simplified counterparts when the jumps of the hybrid system are periodic. A numerical simulation is finally given to illustrate the results.
\end{abstract}


\begin{multicols}{2}

\section{Introduction}

\textit{Hybrid systems} are dynamical models that integrate the dynamics of continuous-time systems and those of discrete-time systems, see~\cite{GoeSanTee:09}. Such dynamical models appear in various engineering fields, such as walking robots and mechanical systems with impacts~\cite{Bro:99,MorPanStr:88}, switching circuits~\cite{ref:erickson2007fundamentals}, and neural models~\cite{Buc:88,PikRosKur:01}.
As a result of being studied and analysed by different research communities, hybrid systems have been studied under different terminologies, such as \textit{hybrid automata}~\cite{Hen:96}, \textit{switching systems}~\cite{LibMor:99}, and \textit{impulsive systems}~\cite{BaiSim:89,ref:zavalishchin2013dynamic,ref:liu2019impulsive}. While many control problems regarding hybrid systems have been recently solved (see, \textit{e.g.},~\cite{ref:Goebel2012hybrid,ref:sanfelice2021hybrid} and references therein), others, such as the problem of model reduction, remain open.

Given a dynamical system, the problem of model reduction consists in the construction of a simplified mathematical description with respect to some notion of complexity, while preserving part of the behaviour and properties of the original description. The literature on model reduction for continuous-time and discrete-time systems is rich and the field can be considered mature, see, \textit{e.g.}, \cite{Glo:84,Moo:81,Kim:86,AntBalKanWil:90,SchGra:00,Ant:05,GraVer:06,HinVol:05,WilPer:02,Ast:10,ref:scarciotti2017nonlinear}. However, only few attempts have been made to solve the problem of model reduction for hybrid systems. In particular,~\cite{HabVaS:02} has addressed the problem of reduction on polytopes, and~\cite{MazVinBalBic:08}, based on the concept of abstraction, has presented a different approach by using balanced truncation for the continuous-time subsystems and pseudo-equivalent location elimination for the discrete-time subsystem. Many model reduction methods were also proposed for discrete-time switching systems~\cite{GaoLamWang:06,WuZhe:09,ZhaShiBouWan:08}, switching systems of Markovian type~\cite{ZhaHuaLam:03}, and switching systems with average dwell time~\cite{ZhaShi:08,ZhaBouShi:09}.~\cite{ShaWis:09,ShaWis:11} have presented an approach based on the notion of generalized Gramians to address the problem of model reduction of switching systems,~\cite{ref:peitz2019koopman} has introduced Koopman operator-based model reduction to simplify the control of nonlinear systems by switched-system. Finally, the problem of reducing a hybrid system to a continuous-time model near a periodic orbit has been investigated in~\cite{BurRevSas:11}. 

\vspace{3mm}
\begin{figurehere}
\centering
\includegraphics[width=\columnwidth]{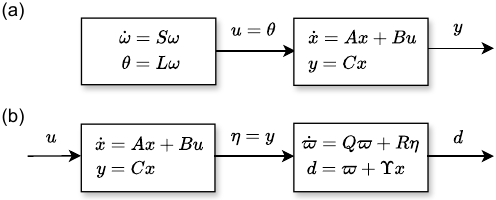}
\caption{Illustration of the direct (a) and swapped (b) interconnections in the moment matching of linear continuous-time systems.}
\label{fig:MMIntLTI}
\end{figurehere}%

This paper addresses the model reduction problem for linear hybrid systems based on the idea of moment matching, which is an interconnection-based model-reduction technique. This constructs reduced-order models that match the steady-state output responses of systems in certain interconnection frameworks, namely the direct and swapped interconnections depicted in Fig.~\ref{fig:MMIntLTI}. \blue{More specifically, consider a complex system to be reduced. The direct interconnection in Fig.~\ref{fig:MMIntLTI}(a) features a signal generator driving the complex system, while the swapped interconnection in  Fig.~\ref{fig:MMIntLTI}(b) consists of the complex system driving an output filter,} see~\cite{ref:scarciotti2024interconnection} for more detail. The work in this paper is originally inspired by~\cite{ScaAst:15c} and~\cite{ref:mao2024model}, in which the classical moment matching approach has been extended to the interconnection of systems with ``explicit'' systems, \textit{i.e.}, systems with dynamics that cannot be modelled by differential equations. While explicit systems can also model hybrid systems,~\cite{ScaAst:15c} and~\cite{ref:mao2024model}
still focus on the model reduction of linear continuous-time systems. This paper extends the model reduction study to hybrid systems by generalising the interconnection frameworks in Fig.~\ref{fig:MMIntLTI} to the hybrid setting. Preliminary results have been presented in~\cite{ref:Scarciotti2016ModelJumps}, in which only the direct interconnection has been considered (see also~\cite{GalSas:15} for the case of periodic jumps). In this paper, we study both the direct and swapped interconnections, and we do not assume periodicity of the jumps. For each interconnection, we first introduce a hybrid characterisation of the steady-state response, herein referred to as the \textit{moment}, and then propose families of linear hybrid reduced-order models capable of matching the steady-state responses associated with each interconnection. By combining the results for each interconnection, we present the design of a reduced-order model capable of matching the steady-state responses for both interconnections (a problem classically called \textit{two-sided moment matching}~\cite{Ionescu2016Two-SidedSystems}). In addition, we also show that the proposed results can be simplified in the periodic case. The results are illustrated by means of a numerical example.

The rest of the paper is organised as follows. In Section~\ref{sec:Prelim} we review the moment matching technique and formulate the hybrid model reduction problem addressed in the paper. Sections~\ref{sec:directMM} and~\ref{sec:swapMM} study the characterisation of the steady-state responses of the direct and swapped interconnections, respectively, and present families of reduced-order models in the sense of moment matching. Section~\ref{sec:2SidedMM} presents the reduced-order model design for two-sided moment matching. All presented results are specialised to the case of periodic jumps in Section~\ref{sec:Periodic}, and illustrated by means of a numerical simulation in Section~\ref{sec:example}. Finally, Section~\ref{sec:concl} provides some concluding remarks.

\textbf{Notation.} We use standard notation. $\mathbb{R}_{\ge0}$ denotes the set of non-negative real numbers, $\mathbb{R}_{>0}$ denotes $\mathbb{R}_{\ge0}\setminus\{0\}$, and $\mathbb{Z}$ denotes the set of all integers. $\mathbb{C}_{<0}$ and $\mathbb{C}_0$ denote the set of complex numbers with strictly negative real part and zero real part, respectively. $\mathbb{D}_{<1}$ denotes the set of complex numbers with modulus strictly smaller than one and $\mathbb{D}_{1}$ the set of complex numbers with modulus equal to one. The symbol $I_n$ denotes the identity matrix of order $n$, $\sigma(A)$ denotes the spectrum of the matrix $A\in\mathbb{R}^{n\times n}$, and $\|A\|$ indicates the induced Euclidean matrix norm. 
Given a function $\omega : \mathbb{R} \times \mathbb{Z} \to \mathbb{R}^{\nu}$ defined on the ``hybrid time domain'' $\mathcal{H}\subset \mathbb{R} \times \mathbb{Z}$, $\dot \omega$ indicates its time derivative, namely $\dot \omega= \frac{\partial }{\partial t}\omega(t, j)$, whereas $\omega^+$ indicates its forward discrete shift, namely $\omega^+=\omega(t, j+1)$. When a function $\Pi(\cdot,\cdot)$ defined on $\mathcal{H}$ is indicated as $\Pi_j$, we are denoting the value $\Pi(t_j,j)$, \textit{i.e.} the value of $\Pi(t, j)$ at the time the $j$-th jump occurred. Given $k - p + 1$ matrices $X_j \in \mathbb{R}^{n \times n}$ for $j=p, (p+1), \dots, (k-1), k$, with any integer $p \leq k$, $\prod_{\overrightarrow{j=p}}^{k} X_j$ and $\prod_{\overleftarrow{j=p}}^{k} X_j$ indicate the products $X_{p} X_{p+1}\cdots X_k$ and $X_{k} X_{k-1}\cdots X_p$, respectively. In addition, if $p > k$, $\prod_{\overrightarrow{j=p}}^{k} X_j = \prod_{\overleftarrow{j=p}}^{k} X_j = I_n$

\section{Preliminaries and problem formulation}
\label{sec:Prelim}

In this section we introduce a few basic definitions related to the steady-state responses, and we formulate the moment matching problem addressed in the paper.

\subsection{On the Notion of Steady-state Response}
Consider a system of interest. We introduce the notion of steady-state response, starting with representations of the state trajectories.

\begin{definition}\textbf{\cite{ZadDes:63,KalFalArb:69}}
\label{NOL-defExp}
Let $x(t) \in \mathbb{R}^n$ be the state of a dynamical system $\Sigma$ and $u(t) \in \mathbb{R}^m$ be the $m$-dimensional input of $\Sigma$. Let $t_0$ and $x_0=x(t_0)$ be the initial time and the initial state, respectively. If there exists a function $\phi: \mathbb{R} \times \mathbb{R} \times \mathbb{R}^n \times \mathbb{R}^m \to \mathbb{R}^n$ such that
\begin{equation}
\label{ch7-eq4a}
x(t)=\phi(t,t_0,x_0, u_{[t_0, t)}),
\end{equation}
for all $t\ge t_0$, we call equation (\ref{ch7-eq4a}) the \textit{representation in explicit form}, or the \textit{explicit model}, of $\Sigma$.\\
Assume $\phi(t,t_0,x_0,u)$ has a continuous derivative with respect to $t$ for every $t_0$, $x_0$ and $u$, and there exists a function $f: \mathbb{R}^n \times \mathbb{R}^m \to \mathbb{R}^n$, continuous for each $t$ over $\mathbb{R}^n \times \mathbb{R}^m$, such that
\begin{equation}
\label{ch7-eq4b}
\dot x = f(x,u).
\end{equation}
We call the differential equation~(\ref{ch7-eq4b}) the \textit{representation in implicit form}, or the \textit{implicit model}, of $\Sigma$.
\end{definition}

We refer the reader to~\cite{ref:isidori2008steady} for a formal definition of steady-state and we stress that the steady-state response can be
defined, under a few hypotheses, for systems in explicit form, \textit{i.e.} it is not necessary that the steady-state response be the solution of a differential equation. 

\subsection{Review of Moment Matching}\label{subsec:MMreview}
We now recall the basic concepts of moment matching. Consider a \blue{linear continuous-time minimal system} described by the equations
\begin{equation}
\label{equ:LTIsystem}
\begin{array}{l}
\dot x =Ax+Bu,\\
y=Cx,
\end{array}
\end{equation}
with $x(t)\in \mathbb{R}^n$, \blue{$u(t)\in \mathbb{R}^m$, $y(t)\in \mathbb{R}^p$, $A\in \mathbb{R}^{n\times n}$, $B\in \mathbb{R}^{n\times m}$ and $C\in \mathbb{R}^{p\times n}$}. Then the moment of system~\eqref{equ:LTIsystem} is defined as follows.

\begin{definition}\textbf{\cite{ref:scarciotti2024interconnection}} \label{def:moment}
    Consider system~(\ref{equ:LTIsystem}) and any two matrices $S \in \mathbb{R}^{\nu \times \nu}$ and $Q \in \mathbb{R}^{\nu \times \nu}$ with their spectrums satisfying $\sigma(S) \cap \sigma(A) = \emptyset$ and $\sigma(Q) \cap \sigma(A) = \emptyset$. Let matrices \blue{$L \in \mathbb{R}^{m \times \nu}$ and $R \in \mathbb{R}^{\nu \times p}$} be such that the pair $(S, L)$ is observable and the pair $(Q, R)$ is reachable. Then
    the matrices $C\Pi$ and $\Upsilon B$ are called the \textit{moments} of system~(\ref{equ:LTIsystem}) at $(S,\, L)$ and $(Q,\, R)$ respectively, where $\Pi \in \mathbb{R}^{n \times \nu}$ and $\Upsilon \in \mathbb{R}^{\nu \times n}$ are the unique solutions of the Sylvester equations
    \begin{subequations}\label{equ:MomSylv}
    \begin{align}
        \Pi S &= A\Pi + BL, \label{equ:SylvPiLTI} \\
        Q\Upsilon &= \Upsilon A + RC.\label{equ:SylvUpsilonLTI}
    \end{align}
    \end{subequations}
\end{definition}




\begin{remark}\label{rmk:momentFreq}
    The moment is classically defined in terms of the transfer function of the system. More specifically, \blue{consider a single-input single-output system~\eqref{equ:LTIsystem}, \textit{i.e.},$m = p = 1$.} The $k$-th moment of this system at an interpolation point $\bar{s} \in \mathbb{C} \setminus \sigma(A)$ is defined as the $k$-th coefficient of the Laurent series expansion of its transfer function at $\bar{s}$~\cite{Ant:05}.~\cite{GalVanVDo:04a} and~\cite{Ast:10} have shown that the moments at the interpolation points coinciding with the eigenvalues of the matrices $S$ and $Q$ are uniquely determined by the elements of $C\Pi$ and $\Upsilon B$. \blue{The moment matching problem for multiple-input multiple-output (MIMO) systems has been studied by~\cite{ref:shakib2023time}.}
\end{remark}

While the moment is originally defined in the Laplace domain, its relation with the Sylvester equations~\eqref{equ:MomSylv}, inspired by the regulator theory in~\cite{Hua:04}, provides a bridge between the moments and the steady-state output responses of certain interconnection structures, namely the direct and swapped interconnections, as displayed in Fig.~\ref{fig:MMIntLTI} and summarised as follows. 

\begin{theorem}\textbf{\cite{Ast:10a,ref:padoan2017geometric,ref:mao2024model}}. \label{thm:MMLTIss}
Consider system~(\ref{equ:LTIsystem}) and suppose $\sigma(A) \subset \mathbb{C}_{<0}$. 
Let $S \in \mathbb{R}^{\nu \times \nu}$ and $Q \in \mathbb{R}^{\nu \times \nu}$ be any two 
matrices with simple eigenvalues and disjoint spectra satisfying $\sigma(S) \subset \mathbb{C}_{0}$ and $\sigma(Q) \subset \mathbb{C}_{0}$. Let $L \in \mathbb{R}^{m \times \nu}$ and $R \in \mathbb{R}^{\nu \times p}$ be any two matrices such that the pair $(S, L)$ is observable and the pair $(Q, R)$ is reachable. Then the following statements hold
\begin{itemize}
    \item The moment at $(S, L)$ is in a one-to-one relation with the steady-state response of the output $y$ of the direct interconnection between system~(\ref{equ:LTIsystem}) and the signal generator
    \begin{equation}
    \label{eq:genS}
        \dot{\omega} =S \omega, \quad \theta=L \omega,
    \end{equation}
    via $u = \theta$ (Fig.~\ref{fig:MMIntLTI}(a)), provided $(S,\, \omega(0))$ is reachable.
    \item The moment at $(Q, R)$ admits a one-to-one relation with the steady-state response of the output $\varpi$ of the swapped interconnection between system~(\ref{equ:LTIsystem}) and the filter
    \begin{equation}
    \label{eq:filQ}
        \dot{\varpi} = Q \varpi+R \eta,
    \end{equation}
   via $\eta = y$ (Fig.~\ref{fig:MMIntLTI}(b)), for $x(0) = \varpi(0) = 0$ and any non-zero signal $u$ that exponentially decays to zero.
\end{itemize}
\end{theorem}

Theorem~\ref{thm:MMLTIss} provides an ``interconnection-based'' interpretation of moment matching, suggesting the possibility of extending the method to general classes of systems, see, \textit{e.g.},~\cite{ref:scarciotti2017nonlinear,ref:scarciotti2024interconnection} and references therein. In this paper, we also rely on this ``interconnection-based'' interpretation to study the moment matching problem for linear hybrid systems.

As a result of Definition~\ref{def:moment} and Theorem~\ref{thm:MMLTIss}, one can easily show that the family of models
\begin{equation}\label{equ:ROMSL}
    \dot{\xi} = (S-G L) \xi+G u, \quad \psi=C \Pi \xi,
\end{equation}
matches the moment of system~\eqref{equ:LTIsystem} at $(S, L)$ for any $G \in \mathbb{R}^{\nu \times m}$ such that $\sigma(S - GL) \cap \sigma(S) = \emptyset$. Likewise, the family of models
\begin{equation}\label{equ:ROMQR}
\dot{\xi}=(Q-R H) \xi+\Upsilon B u, \quad \psi=H \xi,
\end{equation}
matches the moment of system~\eqref{equ:LTIsystem} at $(Q, R)$ for any $H \in \mathbb{R}^{p \times \nu}$ such that $\sigma(Q - RH) \cap \sigma(Q) = \emptyset$. Moreover, when $\sigma(S) \cap \sigma(Q) = \emptyset$, designing a model that matches the moments of~(\ref{equ:LTIsystem}) at $(S,\,L)$ and $(Q,\, R)$ simultaneously (and therefore matching the steady state output responses for both interconnections in Fig.~\ref{fig:MMIntLTI}(a) and (b)) is called two-sided moment matching.
This matching can be achieved, for example, by model~\eqref{equ:ROMSL} with $G = (\Upsilon \Pi)^{-1} \Upsilon B$, or by model~\eqref{equ:ROMQR} with $H = C \Pi(\Upsilon \Pi)^{-1}$, provided that $\Upsilon \Pi$ is invertible, see~\cite{Ionescu2016Two-SidedSystems} for more detail.

\subsection{Problem Formulation}
We now formulate the model reduction problem for a linear hybrid system. To this end, we first consider the hybrid time domain that underlies the evolution of all hybrid systems studied in this paper, namely,
\begin{equation}\label{equ:hybTimeDomain}
\mathcal{H}:=\bigcup_{j=-\infty}^{+\infty}\left(\left[t_j, t_{j+1}\right], j\right) \in \mathbb{R} \times \mathbb{Z},
\end{equation}
for some infinitely many strictly increasing time instants $-\infty < \dots < t_{-1} < t_0 <  t_1 < \dots < +\infty$ with $t_0 = 0$, which are not assumed to be equally spaced but to satisfy the following assumption.
\begin{assumption}\label{asmp:jumpSpacing}
    Consider the countable set $\mathcal{J} := \{t_j\}_{j = -\infty}^{+\infty}$. There exists positive constants $\underline{\delta}$ and $\bar{\delta}$ such that the spacing between two successive jump instants, that is, $\delta_{j} = t_{j+1} - t_{j}$, satisfies
    \begin{equation}\label{equ:JumpBounds}
        \underline{\delta}:=\inf_j \delta_j>0 \quad \text{and} \quad \bar{\delta}:=\sup _j \delta_j<\infty
    \end{equation}
    for any $j \in \mathbb{Z}$.
\end{assumption}
Assumption~\ref{asmp:jumpSpacing} guarantees that $t_j \to \pm \infty$ as $j \to \pm \infty$ and the number of jumps within any finite time interval is finite.
\blue{Physically, this assumption is satisfied if the system jumps infinitely many times with a uniform minimum dwell-time between any two successive jumps. Given such a hybrid time domain $\mathcal{H}$ in~\eqref{equ:hybTimeDomain}, 
we denote two sets $\mathcal{S}_c, \mathcal{S}_d \!\subset\! \mathcal{H}$ as
\begin{equation*}
\begin{array}{rl}
    \mathcal{S}_c &:= \{(t, j): t \in [t_j, t_{j+1}]\}, \\
    \mathcal{S}_d &:= \{(t, j): t = t_{j+1}\},
\end{array}
\end{equation*}
which we herein call the flow and jump sets, respectively.}

\begin{remark}
    The jump time set $\{t_j\}_{j=-\infty}^{+\infty}$ in the hybrid time domain~\eqref{equ:hybTimeDomain} can either be an \textit{a priori} known set or a set that depends on the state transition of certain system dynamics. We will illustrate this state-dependent case with an example in Section~\ref{sec:example}.
\end{remark}

This paper focuses on the model reduction of \blue{a linear hybrid} system described by the equations
\begin{equation}
\label{equ:HybFOM}
\Sigma:\left\{\,\begin{array}{lclc}
\dot x &=&A_cx+B_cu_c, & \quad \blue{(t, j) \in \mathcal{S}_c},\\
x^+ &=&A_dx+B_du_d, & \quad \blue{(t, j) \in \mathcal{S}_d},\\
y &=&Cx,
\end{array}\right.
\end{equation}
with $x(t, j)\in \mathbb{R}^n$, \blue{$u_{c}(t, j)\in \mathbb{R}^m$, $u_{d}(t, j)\in \mathbb{R}^m$, $y(t, j)\in \mathbb{R}^p$, $A_{c}\in \mathbb{R}^{n\times n}$, $A_{d}\in \mathbb{R}^{n\times n}$, $B_{c}\in \mathbb{R}^{n\times m}$, $B_{d}\in \mathbb{R}^{n\times m}$ and $C\in \mathbb{R}^{p\times n}$}. System~(\ref{equ:HybFOM}) flows and jumps according to the hybrid time domain $\mathcal{H}$ in~\eqref{equ:hybTimeDomain}. \blue{More specifically, it flows whenever $(t, j) \in \mathcal{S}_c$, and jumps whenever $(t, j)\in \mathcal{S}_d$}. 

\begin{remark}
\blue{The hybrid system~\eqref{equ:HybFOM} can be expressed as the conventional formulation in~\cite{ref:Goebel2012hybrid} by introducing two clock variables $\tau(t,j) \in \mathbb{R}$ and $k(t,j) \in \mathbb{Z}$ satisfying
    \begin{equation}\label{equ:counter}
        \begin{array}{rlrll}
            \dot{\tau}\,\,\,&= 1 &\qquad \dot{k} \,\,\,&= 0, &\quad (\tau, k) \in \hat{\mathcal{S}}_c, \\
            \tau^+ \!&= \tau &\qquad k^+ \!&= k + 1, &\quad (\tau, k) \in \hat{\mathcal{S}}_d, 
        \end{array}
    \end{equation}
    where $\hat{\mathcal{S}}_c$ and $\hat{\mathcal{S}}_d$ are flow and jump sets defined according to the set $\mathcal{J} = \{t_j\}_{j = -\infty}^{+\infty}$. Then the flow and jump conditions in~\eqref{equ:HybFOM} can be replaced by those of~\eqref{equ:counter}, and system~\eqref{equ:HybFOM} can be written as~\cite[Equation~(1.2)]{ref:Goebel2012hybrid} by involving the dynamics of the clock variables~\eqref{equ:counter}. Since $\tau$ and $k$ play a similar role as $t$ and $j$, and do not simplify any formulation when elements in $\mathcal{J}$ are not equally spaced, we omit these variables for brevity of this paper. 
    }
\end{remark}

To address the hybrid model reduction problem by moment matching, we study a hybrid generalisation of the interconnection structures illustrated in Fig.~\ref{fig:MMIntLTI}. To start with, the direct interconnection in the hybrid setting, as shown in Fig.~\ref{fig:MMIntHyb}(a), is formed by~\eqref{equ:HybFOM} connected to a hybrid signal generator of the form
\begin{equation}
\label{equ:HybGen}
\left\{\,\begin{array}{lclrl}
\dot \omega &=& S \omega, \quad &\theta_c\,\, =L_c \omega, & \quad \blue{(t, j) \in \mathcal{S}_c}, \\
\omega^+ &=& J \omega,  \quad &\theta_d\,\, =L_d \omega, & \quad \blue{(t, j) \in \mathcal{S}_d},
\end{array}\right.
\end{equation}
with $\omega(t, j)\in \mathbb{R}^\nu$, $\theta_{c}(t, j)\in \mathbb{R}$, $\theta_{d}(t, j)\in \mathbb{R}$, $S\in \mathbb{R}^{\nu\times \nu}$, $J\in \mathbb{R}^{\nu\times \nu}$, \blue{$L_{c}\in \mathbb{R}^{m \times \nu}$, $L_{d}\in \mathbb{R}^{m \times \nu}$}. To study the model reduction problem for this direct interconnection, the following assumption is introduced.

\begin{assumption}\label{asmp:genSJ}
The trajectories of the signal generator~(\ref{equ:HybGen}) are bounded forward and backward in time. The pairs $(S, L_c)$ and $(J, L_d)$ are observable. The matrix $J$ is such that $0\not \in \sigma(J)$.
\end{assumption}
%

\begin{figurehere}
\centering
\includegraphics[width=\columnwidth]{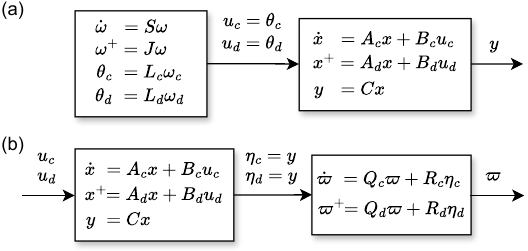}
\caption{Illustration of the direct (a) and swapped (b) interconnections in the moment matching of linear hybrid systems.}
\label{fig:MMIntHyb}
\end{figurehere}%

In Assumption~\ref{asmp:genSJ}, the boundedness requirement yields that the signals $\theta_c$ and $\theta_d$ generated by system~\eqref{equ:HybGen} are bounded for all $(t, j) \in \mathcal{H}$. This boundedness, together with the observability requirement of Assumption~\ref{asmp:genSJ}, can be seen as an extension of the assumptions on $\sigma(S) \subset \mathbb{C}_0$ and the observability of the pair $(S, L)$ in Theorem~\ref{thm:MMLTIss}, as the requirement $\sigma(S) \subset \mathbb{C}_0$ guarantees forward and backward boundedness of the solutions generated by system~\eqref{eq:genS}. The invertibility of $J$ is an extra assumption used to guarantee time-reversibility of system~\eqref{equ:HybGen}.
\red{With this assumption,} the first model reduction problem to be addressed in this paper is formulated as follows.

\begin{problem}\label{prob:directMM}
\textbf{(Model Reduction for Direct Interconnection).}
Consider the interconnection in Fig.~\ref{fig:MMIntHyb}(a). Suppose Assumption~\ref{asmp:genSJ} holds and the zero equilibrium of system~$\Sigma$ in~(\ref{equ:HybFOM}) is exponentially stable. Determine a linear hybrid system~$\Sigma^\prime$ of order $\nu<n$ such that $\psi_{ss}(t, j) = y_{ss}(t, j)$ for all $(t, j) \in \mathcal{H}$ when subject to the same inputs $u_c = \theta_c$ and $u_d = \theta_d$, where $\psi_{ss}$ and $y_{ss}$ stands for the steady-state output responses of the systems~$\Sigma$ in~\eqref{equ:HybFOM} and~$\Sigma^\prime$, respectively.
\end{problem}

We now formulate the hybrid model reduction problem for the swapped interconnection. As depicted in Fig.~\ref{fig:MMIntHyb}(b), this interconnection is composed by the cascade of system~\eqref{equ:HybFOM} with a filter of the form \begin{equation}\label{equ:HybFilt}
    \left\{\,\begin{array}{lcll}
    \dot{\varpi} &=& Q_c \varpi + R_c \eta_c , &\quad\blue{(t, j) \in \mathcal{S}_c}, \\
    \varpi^+ &=& Q_d \varpi + R_d \eta_d, &\quad\blue{(t, j) \in \mathcal{S}_d},
    \end{array}\right.
\end{equation}
where $\varpi(t, j) \in \mathbb{R}^{\nu}$, \blue{$\eta_c(t, j) \in \mathbb{R}^p$, $\eta_d(t, j) \in \mathbb{R}^p$, $Q_c \in \mathbb{R}^{\nu \times \nu}$, $Q_d \in \mathbb{R}^{\nu \times \nu}$, $R_{c}\in \mathbb{R}^{\nu \times p}$, and $R_{d}\in \mathbb{R}^{\nu \times p}$}. Similarly to Assumption~\ref{asmp:genSJ}, we introduce another assumption, which is also a generalisation of the assumptions used in Theorem~\ref{thm:MMLTIss}.

\begin{assumption}\label{asmp:filtQR}
The trajectories of the system~(\ref{equ:HybFilt}) with $\eta_c \equiv \eta_d \equiv 0$ are bounded forward and backward in time. The pairs $(Q_c, R_c)$ and $(Q_d, R_d)$ are reachable. The matrix $Q_d$ is such that $0\not \in \sigma(Q_d)$.
\end{assumption}

Similarly to Assumption~\ref{asmp:genSJ}, the boundedness and reachability requirements of Assumption~\ref{asmp:filtQR} can also be regarded as an extension of the assumptions that $\sigma(Q) \subset \mathbb{C}_{0}$ and the pair $(Q, R)$ is reachable in Theorem~\ref{thm:MMLTIss}. 
The assumption on the invertibility of $Q_d$ is introduced to ensure time-reversibility of the filter~\eqref{equ:HybFilt}. 

\begin{problem}\label{prob:swappedMM}\textbf{(Model Reduction for Swapped Interconnection).}
Consider the interconnection in Fig.~\ref{fig:MMIntHyb}(b) with initial conditions $x_0 = 0$ and $\varpi_0 = 0$.
Suppose Assumption~\ref{asmp:filtQR} holds and the zero equilibrium of system~$\Sigma$ in~(\ref{equ:HybFOM}) is exponentially stable. For any non-zero signals $u_c$ and $u_d$ exponentially decaying to zero, determine a linear hybrid system~$\Sigma^\prime$ of order $\nu<n$ such that the steady-state response $\varpi_{ss}$ of the filter, when subject to $\eta_c = \psi$ and $\eta_d = \psi$, coincides with the steady-state response of the filter in the interconnection in Fig.~\ref{fig:MMIntHyb}(b).
\end{problem}

Finally, we formulate the two-sided moment matching problem for hybrid systems as follows.

\begin{problem}\label{prob:twoSidedMM}\textbf{(Model Reduction for Two-sided Interconnection).}
    Consider system~$\Sigma$ in~\eqref{equ:HybFOM} with the signal generator~\eqref{equ:HybGen} and the filter~\eqref{equ:HybFilt}. Suppose Assumptions~\ref{asmp:genSJ} and~\ref{asmp:filtQR} hold, and the zero equilibrium of system~$\Sigma$ in~(\ref{equ:HybFOM}) is exponentially stable. Design a linear hybrid system~$\Sigma^\prime$ of order $\nu < n$ such that Problems~\ref{prob:directMM} and~\ref{prob:swappedMM} are simultaneously solved.
\end{problem}

The goal of this paper is to address the three moment matching problems formulated as Problems~\ref{prob:directMM}, \ref{prob:swappedMM}, and~\ref{prob:twoSidedMM}. Towards this end, the reduced-order model to be constructed is also a linear hybrid system which flows and jumps according to the hybrid time domain $\mathcal{H}$ in~\eqref{equ:hybTimeDomain} and takes the form 
\begin{equation}\label{equ:HybROM}
\Sigma^\prime: \left\{\, \begin{array}{lcll}
\dot{\xi} &=& F_c \xi + G_c u_c, &\quad\blue{(t, j) \in \mathcal{S}_c},\\
\xi^+ &=& F_d \xi + G_d u_d, &\quad\blue{(t, j) \in \mathcal{S}_d},\\
\psi \,\,\, &=& H\xi,
\end{array}\right.
\end{equation}
with $\xi(t, j)\in \mathbb{R}^{\nu}$, \blue{$\psi(t, j)\in \mathbb{R}^p$, $F_c(t, j) \in \mathbb{R}^{\nu \times \nu}$, $F_d(t, j) \in \mathbb{R}^{\nu \times \nu}$, $G_{c}(t, j)\in \mathbb{R}^{\nu \times m}$, $G_{d}(t, j)\in \mathbb{R}^{\nu \times m}$, and $H(t, j) \in \mathbb{R}^{p \times \nu}$}. \blue{The dimension $\nu$ of the reduced-order model coincides with that of the signal generator~\eqref{equ:HybGen} and/or the filter~\eqref{equ:HybFilt}, \textit{i.e.}, it depends on the complexity of the system that excites/filters the dynamics of the complex system~\eqref{equ:HybFOM} in the interconnection in Fig.~\ref{fig:MMIntHyb}.} Note also that, depending on the specific model reduction task to be achieved, some parameters of the reduced-order model can be time-varying. The reason for this will be clear later. 

The direct and swapped moment matching problems formulated in Problems~\ref{prob:directMM} and~\ref{prob:swappedMM} are similar to the ``non-smooth'' model reduction problem in~\cite{ScaAst:15c,ref:mao2024model}, which have addressed the moment matching problem for linear continuous-time systems~\eqref{equ:LTIsystem} in the direct interconnection with an explicit signal generator of the form
\begin{equation}\label{equ:explicitGen}
\hat{\omega}(t)=\Lambda(t) \hat{\omega}(0), \qquad u=\hat{L}\hat{\omega},
\end{equation}
where $\hat{\omega} \in \mathbb{R}^{\nu}$, $\hat{L}\in\mathbb{R}^{m \times \nu}$, and $\Lambda(t)\in \mathbb{R}^{\nu\times \nu}$ are such that $\Lambda(0)=I_\nu$, or in the swapped interconnection with the explicit filter of the form
\begin{equation}\label{equ:explicitFilt}
\hat{\varpi}(t)=\Omega(t) \int_0^t \Omega(\vartheta)^{-1} \hat{R} \eta(\vartheta) d \vartheta,
\end{equation}
where $\hat{\varpi} \in \mathbb{R}^{\nu}$, $\Omega(t) \in \mathbb{R}^{\nu \times \nu}$, and $\hat{R} \in \mathbb{R}^{\nu \times p}$ are such that $\Omega(0)= I_\nu$ and $\Omega(-t)=\Omega(t)^{-1}$. Since $\Lambda$ and $\Omega$ are not assumed to be continuous, these explicit systems can also be used to model \blue{dynamics of some hybrid systems. For example, if (and only if) $L_c = L_d$ in~\eqref{equ:HybGen} and $R_d = 0$ in~\eqref{equ:HybFilt}, then systems~\eqref{equ:HybGen} and~\eqref{equ:HybFilt} can be described by~\eqref{equ:explicitGen} and~\eqref{equ:explicitFilt}, respectively} Under certain assumptions, see~\cite{ScaAst:15c}, the steady-state response of system~\eqref{equ:LTIsystem} in the direct interconnection with~\eqref{equ:explicitGen} can be described by $x_{ss}(t)=\Pi_{\infty}(t)\hat{\omega}(t)$ with
\begin{equation}
\label{equ:explicitPi}
\Pi_{\infty}(t) = \left(\int_{-\infty}^t e^{A(t-\vartheta)}B\hat{L}\Lambda(\vartheta) d\vartheta \right) \Lambda(t)^{-1}.
\end{equation}
Meanwhile, the steady-state response of the filter~\eqref{equ:explicitFilt}, in the swapped interconnection with system~\eqref{equ:LTIsystem}, is characterised by
\begin{equation}
\varpi_{ss}(t)=\Omega(t) \lim _{\tilde{t} \rightarrow +\infty} \int_0^{t+\tilde{t}} \Omega(\vartheta)^{-1} \Upsilon_{\infty}(\vartheta) B u(\vartheta) d \vartheta,
\end{equation}
with
\begin{equation}\label{equ:explicitUpsilon}
\Upsilon_{\infty}(t):=\Omega(t) \int_t^{+\infty} \Omega(\vartheta)^{-1} \hat{R} C e^{A(\vartheta-t)} d \vartheta.
\end{equation}
The formulations~\eqref{equ:explicitPi} and~\eqref{equ:explicitUpsilon} are useful for explicitly characterising the time history of the steady-state response in the interconnection structures. However, \blue{unlike~\cite{ScaAst:15c,ref:mao2024model} that aim at reduction of continuous-time systems, when the system to be reduced is also hybrid}, it is preferable to characterise the system dynamics in a hybrid formulation as in~\cite{GalSas:15}, as this allows defining a reduced-order model also in a hybrid form. Note that~\cite{GalSas:15} addressed the direct interconnection moment matching problem for the linear continuous-time system~\eqref{equ:LTIsystem} using a hybrid formulation. However, the result therein is confined to the periodic case and can also be studied using~\cite{ScaAst:15c}.

In the remainder of the paper we first address Problems~\ref{prob:directMM} and~\ref{prob:swappedMM} separately, providing families of reduced-order models for the direct interconnection and the swapped interconnection. By combining the obtained results, we solve the two-sided moment matching problem formulated in Problem~\ref{prob:twoSidedMM}. We finally demonstrate that when the jump instants $\{t_j\}_{j = -\infty}^{+\infty}$ of the hybrid time domain~\eqref{equ:hybTimeDomain} are periodic, all proposed theories for finding the reduced-order model can be simplified.

\section{Moment Matching for the Direct Interconnection}
\label{sec:directMM}

In this section we address the direct interconnection moment matching problem formulated in Problem~\ref{prob:directMM}, starting with generalising the Sylvester equation~\eqref{equ:SylvPiLTI} to the hybrid setting. This solution is helpful for characterising the steady-state response of the hybrid system~\eqref{equ:HybFOM} in the direct interconnection in Fig.~\ref{fig:MMIntHyb}(a)

\begin{lemma} 
\label{lem:PiExist}
Consider systems~(\ref{equ:HybFOM}) and~(\ref{equ:HybGen}). Suppose Assumptions~\ref{asmp:jumpSpacing} and~\ref{asmp:genSJ} hold, and
the zero equilibrium of system~(\ref{equ:HybFOM}) is exponentially stable. Then there exists a unique bounded (steady-state) solution $\Pi(t, j) = \hat{\Pi}(t, j) \in \mathbb{R}^{n \times \nu}$ solving
\begin{equation}
\label{equ:HybPi}
\begin{array}{ccll}
\dot{\Pi} &=& A_c\Pi - \Pi S + B_cL_c, & \quad \blue{(t, j) \in \mathcal{S}_c}, \\
\Pi^+J &=& A_d\Pi+B_dL_d, & \quad \blue{(t, j) \in \mathcal{S}_d},
\end{array}
\end{equation}
and satisfying the boundary condition 
\begin{equation}\label{equ:PiBodCond}
    \hat{\Pi}_{j}=K_{\hat{\Pi},j},
\end{equation}
for all $(t, j) \in \mathcal{H}$, where
\begin{equation}
\label{equ:PikKpi}
\begin{array}{l}
\!\!\!K_{\hat{\Pi},j} = \sum_{i=-\infty}^{j}\!\!\bigg(\!\! \prod_{\overleftarrow{q=i+1}}^{j} \left(A_de^{A_c(t_{q} - t_{q-1})}\right) \\
\!\!\!\times \Big(\! B_d L_d J^{-1} \!+\! A_d \big(\!\int_{t_{i-1}}^{t_{i}} \!\!\!e^{A_c(t_{i}-\vartheta)}B_cL_c e^{S(\vartheta-t_{i})}d\vartheta \big) J^{-1}\!\Big) \\
\!\!\! \times \prod_{\overrightarrow{q=i+1}}^{j} \left(e^{-S(t_{q}-t_{q-1})}J^{-1} \right) \!\! \bigg).
\end{array}
\end{equation}
Moreover, this solution $\hat{\Pi}$ is forward attractive, that is, it is such that for any $\Pi$ solving~\eqref{equ:HybPi} and starting with any initial condition $\Pi(\underline{t}, \underline{j}) \in \mathbb{R}^{n \times \nu}$ with $(\underline{t}, \underline{j}) \in \mathcal{H}$, $\lim_{t+j \to +\infty} \Pi(t, j) - \hat{\Pi}(t, j) = 0$. 
\end{lemma}

\begin{proof}
    We start by showing that~(\ref{equ:HybPi}) admits a bounded solution $\hat{\Pi}$ \red{satisfying~\eqref{equ:PiBodCond}}. We first prove that the solution of the first $\Pi$-equation in~(\ref{equ:HybPi}) \red{whenever $(t, j) \in \mathcal{S}_c$} is
\begin{equation}
\label{equ:PiExpInterval}
\begin{array}{l}
\Pi(t, j)=\left(\vphantom{\int_{t_j}^t }   e^{A_c(t-t_j)}\Pi_j+\right.\\
\qquad\qquad\,\,+\left.\int_{t_j}^t e^{A_c(t-\vartheta)}B_cL_c e^{S(\vartheta-t_j)}d\vartheta\right)e^{-S(t-t_j)}.
\end{array}
\end{equation}
To this end note that $\Pi(t, j)$ in (\ref{equ:PiExpInterval}) is always well-defined since the integral is over a finite interval and contains only the product of exponential matrices. \red{Whenever $(t, j) \in \mathcal{S}_c$}, $\Pi(t, j)$ is differentiable, yielding
\begin{equation*}
\begin{array}{l}
\!\!\dot \Pi e^{S(t-t_j)} \!+ \Pi Se^{S(t-t_j)}=  \!A_c e^{A_c(t-t_j)}\Pi_j\\
\qquad\qquad\!+ B_cL_c e^{S(t-t_j)}+ A_c\!\int_{t_j}^t \!\!\!e^{A_c(t-\vartheta)}B_cL_c e^{S(\vartheta-t_j)}d\vartheta,
\end{array}
\end{equation*}
with the last two terms obtained by the Leibniz rule. Substituting the integral appearing in the last equation with the expression given in (\ref{equ:PiExpInterval}) yields
\begin{equation*}
\begin{array}{l}
\!\!\dot \Pi e^{S(t-t_j)} + \Pi Se^{S(t-t_j)}=  \!A_c e^{A_c(t-t_j)}\Pi_j\\
\qquad\qquad+ B_cL_c e^{S(t-t_j)}\!+\! A_c\left(\Pi e^{S(t-t_j)}-e^{A_c(t-t_j)}\Pi_j    \right),
\end{array}
\end{equation*}
from which, by right-multiplying both sides by $e^{-S(t-t_j)}$, one has 
$$
\dot \Pi  + \Pi S=  B_cL_c + A_c\Pi,
$$
proving that $\Pi$ in (\ref{equ:PiExpInterval}) is the solution of the first equation in (\ref{equ:HybPi}) \red{whenever $(t, j) \in \mathcal{S}_c$}. This suggests that the existence of $\Pi$ solving~\eqref{equ:HybPi} boils down to the existence of the initial condition $\Pi_j$. To study this, we then compute the dynamics of $\Pi$ across the jumps. Writing~(\ref{equ:PiExpInterval}) for $t=t_{j+1}$, replacing the resulting expression of $\Pi(t_{j+1},j)$ by the right-hand side of the second $\Pi$-equation in (\ref{equ:HybPi}), and multiplying on the right for $e^{S(t_{j+1}-t_j)}$, yields
\begin{equation}
\label{equ:PiJump}
\begin{array}{l}
\!\!\Pi_{j+1}Je^{S(t_{j+1}-t_j)}\!=A_de^{A_c(t_{j+1}-t_j)}\Pi_{j}+B_dL_de^{S(t_{j+1}-t_j)}\\
\qquad\qquad\qquad\quad\,\,+A_d\int_{t_j}^{t_{j+1}} \!\!e^{A_c(t_{j+1}-\vartheta)}B_cL_c e^{S(\vartheta-t_j)}d\vartheta.
\end{array}
\end{equation}
First of all note that, since $J$ is invertible by assumption, given $\Pi_{j}$, \textit{i.e.}, the boundary condition of (\ref{equ:PiExpInterval}) for $(t, j) \in [t_j,t_{j+1}] \times \{j\}$, we can uniquely determine $\Pi_{j+1}$, \textit{i.e.}, the boundary condition of (\ref{equ:PiExpInterval}) for $(t, j) \in [t_{j+1},t_{j+2}] \times \{j+1\}$, and also uniquely determine the values of $\Pi$ inside the interval $[t_{j+1},t_{j+2}] \times \{j+1\}$ by~\eqref{equ:PiExpInterval}. Thus, if $\Pi_{j}$ exists for all $j \geq 0$, then the solution to~(\ref{equ:HybPi}) is well-defined. Rewriting equation (\ref{equ:PiJump}) for the interval $[t_{j-1},t_{j}] \times \{j-1\}$ yields
\begin{equation*}
\label{NOL-HybRec0}
\begin{array}{l}
\Pi_{j} =A_de^{A_c(t_{j}-t_{j-1})}\Pi_{j-1}e^{-S(t_{j}-t_{j-1})}J^{-1}+B_dL_dJ^{-1}\\
\quad\,\,\,\quad+A_d\left(\int_{t_{j-1}}^{t_{j}} \!\!\!e^{A_c(t_{j}-\vartheta)}B_cL_c e^{S(\vartheta-t_{j})}d\vartheta\right)J^{-1}.
\end{array}
\end{equation*}
Iterating and replacing the expression for $\Pi_{j-1}$, $\Pi_{j-2}$, $\dots$, we obtain the relation
\begin{equation}
\label{equ:PikExp}
\begin{array}{l}
\Pi_{j} = K_{\hat{\Pi},j} \\
+ 
\prod_{\overleftarrow{i=-\infty}}^{j} \! \left(A_de^{A_c(t_{i}-t_{i-1})}\right)\Pi_{-\infty}\prod_{\overrightarrow{i=-\infty}}^{j} \! \left(e^{-S(t_{i}-t_{i-1})}J^{-1}\!\right)\!,
\end{array}
\end{equation}
with $K_{\hat{\Pi},j}$ as in~\eqref{equ:PikKpi} and $\Pi_{-\infty} \in \mathbb{R}^{n \times \nu}$ an arbitrary value when $j \to -\infty$. Note that the second product sequence in (\ref{equ:PikExp}) is bounded by hypothesis, whereas the first product sequence is converging to zero because 
the zero equilibrium of system~(\ref{equ:HybFOM}) is exponentially stable. The term $K_{\hat{\Pi},j}$ is finite for all $j \in \mathbb{Z}$, because it is the sum of constant terms multiplied on the left by increasingly longer sequences of $A_d e^{A_c(t_q-t_{q-1})}$ (for which the norm exponentially converges to zero). These observations of~\eqref{equ:PikExp}, by the definition in~\cite{ref:isidori2008steady}, implies that $\hat{\Pi}$ satisfying~\eqref{equ:PiBodCond} is the steady-state solution of~\eqref{equ:HybPi} and is bounded for all times as $K_{\hat{\Pi},j}$ is bounded whenever $j \to \pm \infty$. Characterised by the boundary condition~\eqref{equ:PiBodCond}, this $\hat{\Pi}$ is unique as it flows as in~\eqref{equ:PiExpInterval} and jumps as in~\eqref{equ:PiJump}, with both dynamics admitting unique solutions under Assumption~\ref{asmp:jumpSpacing}\footnote{See~\cite[Chapter 2]{ref:Goebel2012hybrid} for more detail about the uniqueness of solutions of hybrid systems.}. To show its forward attractivity, let $E:= \Pi - \hat{\Pi}$. Then by~\eqref{equ:HybPi}, we have
\begin{equation}\label{equ:HybPiError}
    \begin{array}{lcll}
    \dot{E} &=& A_c E - E S, &\quad \blue{(t, j) \in \mathcal{S}_c}, \\
    E^+ &=& A_d E J^{-1}, &\quad \blue{(t, j) \in \mathcal{S}_c}.
    \end{array}
\end{equation}
By a similar argument as the one yielding~\eqref{equ:PikExp}, the evolution of $E$ can be explicitly characterised by $E(t, j) \!=\!e^{A_c(t-t_j)} \hat{E}_j e^{-S(t-t_j)}$ where
\begin{equation*}
\begin{array}{l}
\hat{E}_j := \prod_{\overleftarrow{i=\underline{j}}}^{j}  \left(\!A_de^{A_c(t_{i}-t_{i-1})} \right)E_{\underline{j}}\prod_{\overrightarrow{i=\underline{j}}}^{j} \left( e^{-S(t_{i}-t_{i-1})}J^{-1}\right).
\end{array}
\end{equation*}
for any integer $\underline{j} \leq j$. Then by Assumption~\ref{asmp:genSJ} and the exponential stability of system~(\ref{equ:HybFOM}), $\lim_{t+j \to +\infty} E(t, j) = 0$. Equivalently, starting with any initial condition $\Pi(\underline{t}, \underline{j}) \in \mathbb{R}^{n \times \nu}$ with $(\underline{t}, \underline{j}) \in \mathcal{H}$, we have $\lim_{t+j \to +\infty} \Pi(t, j) - \hat{\Pi}(t, j) = 0$. This concludes the proof.
\end{proof}


We are now ready to characterise the steady-state response of system~\eqref{equ:HybFOM} in the direct interconnection.

\begin{theorem} 
\label{thm:PiSS}
Consider the direct interconnection of system~(\ref{equ:HybFOM}) and the hybrid generator~(\ref{equ:HybGen}) shown in Fig.~\ref{fig:MMIntHyb}(a). Suppose Assumptions~\ref{asmp:jumpSpacing} and~\ref{asmp:genSJ} hold, and
the zero equilibrium of system~(\ref{equ:HybFOM}) is exponentially stable. Then the steady-state response of the hybrid system~\eqref{equ:HybFOM} satisfies $x_{ss} = \hat{\Pi}\omega$, with $\hat{\Pi}$ the unique steady-state solution of~\eqref{equ:HybPi}. 
\end{theorem}
\begin{proof}
By the dynamics of~\eqref{equ:HybFOM} and~\eqref{equ:HybPi}, we have
\begin{equation}
\label{NOL-z1SD}
\begin{array}{l}
\!\!\!\dot{\overbrace{x-\Pi\omega}}=A_cx+B_cL_c\omega-\dot \Pi \omega - \Pi S\omega\\
\qquad\,\,\,\,\,\,=A_cx+B_cL_c\omega-A_c\Pi\omega + \Pi S\omega - B_cL_c \omega \!-\! \Pi S\omega\\
\qquad\,\,\,\,\,\,=A_c(x-\Pi\omega)
\end{array}
\end{equation}
\red{whenever $(t, j) \in \mathcal{S}_c$}, and 
\begin{equation}
\label{NOL-z2SD}
\begin{array}{rl}
\left(x-\Pi\omega\right)^+&\!=A_d x + B_dL_d \omega - \Pi^+J\omega\\
&\!=A_d x + B_dL_d \omega - A_d\Pi\omega - B_dL_d\omega\\
&\!=A_d\left(x-\Pi\omega\right)
\end{array}
\end{equation}
\red{whenever $(t, j) \in \mathcal{S}_d$}. The response of the system of equations~(\ref{NOL-z1SD})-(\ref{NOL-z2SD}) can be written explicitly, namely
$$
x(t, j)-\Pi(t, j)\omega(t, j)=\varepsilon(t, j),
$$
with
$$
\varepsilon(t, j) = e^{A_c(t-t_{j})} \prod_{\overleftarrow{i=1}}^{j}\left(A_de^{A_c(t_i-t_{i-1})}\right)\left(x_0-\Pi_{0}\omega_0\right).
$$
As a consequence, the state response of the interconnection of system~(\ref{equ:HybFOM}) with the hybrid generator~(\ref{equ:HybGen}) is
$$
x(t, j)=\Pi(t, j)\omega(t, j)+\varepsilon(t, j).
$$
Finally, note that the term $\varepsilon$ describes a transient response which decays to zero since 
the zero equilibrium of system~(\ref{equ:HybFOM}) is exponentially stable, proving that $\lim_{t+j \to +\infty} x_{ss}(t, j) - \Pi(t, j)\omega(t, j) = 0$. Since $\hat{\Pi}$, by Lemma~\ref{lem:PiExist}, is the unique steady-state solution of~\eqref{equ:HybPi} and is forward attractive, we finally obtain $x_{ss} = \hat{\Pi}\omega$.
\end{proof}

\begin{remark}\label{rmk:SylvGeneralise}
The steady-state solution $\hat{\Pi}$ in Lemma~\ref{lem:PiExist} is a hybrid generalisation of the solution of the standard Sylvester equation~\eqref{equ:SylvPiLTI}. To see this, if the system~\eqref{equ:HybFOM} and the generator~\eqref{equ:HybGen} are purely continuous-time systems without jumps, the boundary condition~(\ref{equ:PikKpi}) with $j = 0$ reduces to
$$
\hat{\Pi}_{0}= K_{\hat{\Pi},0} = \int_{-\infty}^0 e^{-A_c\vartheta}B_cL_c e^{S\vartheta}d\vartheta,
$$
which is the solution to~\eqref{equ:SylvPiLTI} with $A_c = A$, $B_c = B$ and $L_c = L$. Moreover, if system~\eqref{equ:HybFOM} is continuous-time, \textit{i.e.}, $A_d = I_n$ and $B_d = 0$, and the generator~\eqref{equ:HybGen} has $L_d = L_c$, then it is possible to reformulate the steady-state solution $\hat{\Pi}$ of~\eqref{equ:HybPi} by~\eqref{equ:explicitPi}, see~\cite{ScaAst:15c} for more detail. This generalised Sylvester equation, explicitly characterised by~\eqref{equ:explicitPi}, has also been used for the non-smooth output regulation problem~\cite{ref:niu2024adaptive,ref:niu2025output}.
\end{remark}

By Lemma~\ref{lem:PiExist} and Theorem~\ref{thm:PiSS}, the steady-state output $y_{ss}$ of system~\eqref{equ:HybFOM} in the direct interconnection can be characterised by 
\begin{equation}\label{equ:yss}
   y_{ss} = C \hat{\Pi} \omega, 
\end{equation}
with $\hat{\Pi}$ the unique steady-state solution of~\eqref{equ:HybPi}. This one-to-one characterisation leads to the following definition.
\begin{definition}\label{def:HybMomentDirect}\textbf{\blue{(Moment of Hybrid System in Direct Interconnection).}} 
Consider the direct interconnection of system~(\ref{equ:HybFOM}) and the hybrid generator~(\ref{equ:HybGen}). Suppose Assumptions~\ref{asmp:jumpSpacing} and~\ref{asmp:genSJ} hold, and the zero equilibrium of system~(\ref{equ:HybFOM}) is exponentially stable. Then the function $C \hat{\Pi}$ is called the \textit{moment} of system~\eqref{equ:HybFOM} at the tuple $(S, L_c, J, L_d, \mathcal{H})$, where $\hat{\Pi}$ is the uniqe steady-state solution to~\eqref{equ:HybPi}.
\end{definition}

Then, the construction of a family of reduced-order models that satisfy the direct interconnection moment matching in Problem~\ref{prob:directMM} can be given as follows.

\begin{theorem}
\label{thm:directROM}
Consider system~(\ref{equ:HybFOM}) and the signal generator~(\ref{equ:HybGen}). Suppose Assumptions~\ref{asmp:jumpSpacing} and~\ref{asmp:genSJ} hold, and the zero equilibrium of system~(\ref{equ:HybFOM}) is exponentially stable. Then the system described by the equations
\begin{equation}\label{equ:directROM}
    \Sigma^{\prime}:\left\{\begin{array}{lcll}
    \dot{\xi} &=& (S-G_cL_c) \xi + G_c u_c, &\quad \blue{(t, j) \in \mathcal{S}_c},\\
    \xi^+ &=&(J-G_dL_d) \xi + G_d u_d, &\quad \blue{(t, j) \in \mathcal{S}_d},\\
    \psi &=&C\hat{\Pi}\xi, 
    \end{array}\right.
\end{equation}
with $\hat{\Pi}$ the unique steady-state solution of (\ref{equ:HybPi}), is a \textit{model of system~(\ref{equ:HybFOM}) matching the moment at the tuple $(S, L_c, J, L_d, \mathcal{H})$}, for any $G_{c}$ and $G_{d}$ such that
the zero equilibrium of system~(\ref{equ:directROM}) is exponentially stable. Moreover, system~(\ref{equ:directROM}) is a \textit{reduced order model of system~(\ref{equ:HybFOM}) matching the moment at $(S, L_c, J, L_d, \mathcal{H})$} and solves Problem~\ref{prob:directMM} if $\nu<n$.
\end{theorem}
\begin{proof}
Under the assumptions of this theorem and by Theorem~\ref{thm:PiSS}, the direct interconnection of system~(\ref{equ:directROM}) and the signal generator~(\ref{equ:HybGen}) has a steady-state output response described by $C\hat{\Pi} \hat{P}\omega$, with $P(t, j) = \hat{P}(t, j) \in \mathbb{R}^{\nu \times \nu}$ the bounded solution of
\begin{equation}
\label{equ:HybP}
\begin{array}{ccll}
\dot{P} &=& (S-G_cL_c) P - P S + G_cL_c, &\quad \blue{(t, j) \in \mathcal{S}_c},\\
P^+J &=& (J-G_dL_d) P + G_d L_d, &\quad \blue{(t, j) \in \mathcal{S}_d}.
\end{array}
\end{equation}
We see that \(\hat{P}(t, j) = I_\nu \) is the (unique steady-state) solution to~\eqref{equ:HybP} for all $(t, j) \in \mathcal{H}$. As a consequence the steady-state output response $\psi_{ss}$ of system~(\ref{equ:directROM})  becomes $C\hat{\Pi}\omega$, \textit{i.e.} the family of models~(\ref{equ:directROM}) solves Problem~\ref{prob:directMM}.
\end{proof}

\begin{remark}\label{rmk:PiCompute}
    From a computational point of view, the boundary condition~\eqref{equ:PikKpi} that characterises $\hat{\Pi}$ as in~\eqref{equ:PiBodCond} is hard to obtain explicitly, but can be approximated with arbitrary precision by considering a finite number of terms in~(\ref{equ:PikKpi}). \blue{Nevertheless, it is not necessary to explicitly compute~\eqref{equ:PikKpi} to derive the trajectories of $\hat{\Pi}$}. Since the steady-state solution $\hat{\Pi}$ is forward attractive for all initial conditions $\Pi_0 \in \mathbb{R}^{n \times \nu}$, any solution to~\eqref{equ:HybPi} converges to $\hat{\Pi}$, with the effect that the initial condition exponentially disappears. Consequently, the reduced-order model~\eqref{equ:directROM} still achieves matching if $\hat{\Pi}$ is replaced by any $\Pi$, as the initial condition of $\Pi$ does not affect the asymptotic relation $\lim_{t+j \to +\infty} y_{ss}(t, j) - C\Pi(t, j)\omega(t, j) = 0$, but just affects the speed of its convergence.
\end{remark}

The matching of the steady-state required to solve Problem~\ref{prob:directMM}, achieved by Theorem~\ref{thm:directROM}, is equivalent to matching the moment of  system~\eqref{equ:HybFOM}. More specifically, by Definition~\ref{def:HybMomentDirect}, the moment of system~\eqref{equ:HybROM} satisfies $C\hat{\Pi} \hat{P} = C\hat{\Pi}$, and consequently solves Problem~\ref{prob:directMM}.
Note that the vectors $G_c$ and $G_d$ are free parameters and can be used to achieve additional properties for the family of reduced order models~(\ref{equ:directROM}), see~\cite{Ast:10}. In this paper we show that these free parameters can be selected to achieve two-sided moment matching.

\section{Moment Matching for the Swapped Interconnection}\label{sec:swapMM}
In this section we study the swapped interconnection moment matching problem formulated in Problem~\ref{prob:swappedMM}. To this end, similarly to the direct interconnection in Section~\ref{sec:directMM}, by relying on the solution to another hybrid generalisation of a certain Sylvester equation, we characterise the steady-state response of the swapped interconnection in Fig.~\ref{fig:MMIntHyb}(b), and then find the reduced-order model that matches such a steady state.

\begin{lemma}\label{lem:UpsilonExist}
    Consider systems~(\ref{equ:HybFOM}) and~(\ref{equ:HybFilt}). Suppose Assumptions~\ref{asmp:jumpSpacing} and~\ref{asmp:filtQR} hold, and the zero equilibrium of system~(\ref{equ:HybFOM}) is exponentially stable. If $A_d$ is invertible, there exists a unique bounded (steady-state) solution $\Upsilon(t, j) = \hat{\Upsilon}(t, j) \in \mathbb{R}^{\nu \times n}$ solving
    \begin{equation}
    \label{equ:HybUpsilon}
    \begin{array}{ccll}
    \dot{\Upsilon} &=& Q_c\Upsilon - R_c C -\Upsilon A_c, &\quad \blue{(t, j) \in \mathcal{S}_c},\\
    \Upsilon^+A_d &=& Q_d\Upsilon - R_d C, &\quad \blue{(t, j) \in \mathcal{S}_d}. 
    \end{array}
    \end{equation}
    for all $(t, j) \in \mathcal{H}$. Moreover, this solution $\hat{\Upsilon}$ satisfies the boundary condition 
    \begin{equation}\label{equ:UpsilonBodCond}
    \hat{\Upsilon}_j = K_{\hat{\Upsilon}\!, j},
    \end{equation}
    for all $j \in \mathbb{Z}$, where
    \begin{equation}
    \label{equ:UpsilonkKupsilon}
    \begin{array}{l}
        K_{\hat{\Upsilon}\!, j} = \\
         \sum_{i=j}^{+\infty}\!\!\bigg(\!\!\!  \prod_{\overrightarrow{q=j}}^{i-1} \!\!\left(\!e^{Q_c(t_{q}-t_{q+1})}\!Q_d^{-1}\!\right)\!\!\Big(\!\!\int_{t_{i}}^{t_{i+1}} \!\!e^{Q_c(t_i-\vartheta)}\!R_c C e^{A_c(\vartheta-t_i)}\!d\vartheta\\
         -\, e^{Q_c(t_i-t_{i+1})}Q_d^{-1} R_d C e\!^{A_c(t_{i+1}-t_i)} \!\Big)\!\prod_{\overleftarrow{q=j}}^{i-1}\! \left(\!A_de^{A_c(t_{q+1}-t_{q})}\right)\!\!\!\bigg),
    \end{array}
    \end{equation}
    and is backward attractive, that is, $\hat{\Upsilon}$ is such that for any $\Upsilon$ solving~\eqref{equ:HybUpsilon} and starting with any initial condition $\Upsilon(\bar{t}, \bar{j}) \in \mathbb{R}^{\nu \times n}$ with $(\bar{t}, \bar{j}) \in \mathcal{H}$, it holds that
    \begin{equation}\label{equ:UpsilonBackAttract}
    \lim_{t+j \to -\infty} \Upsilon(t, j) - \hat{\Upsilon}(t, j) = 0.
    \end{equation}
\end{lemma}
\begin{proof}
We first show the existence of a bounded solution to~\eqref{equ:HybUpsilon}. \red{Whenever $(t, j) \in \mathcal{S}_c$}, a proof analogous to that of~\eqref{equ:PiExpInterval} shows that the solution~$\Upsilon$ can be expressed as
\begin{equation}
\label{equ:UpsilonExpInterval}
\begin{array}{l}
\Upsilon(t, j)=
\left(\vphantom{\int_{t_j}^t } e^{Q_c(t-t_j+1)}\Upsilon(t_{j+1}, j)+\right.\\
\qquad+\left.\int_{t}^{t_{j+1}} \!\!\!e^{Q_c(t-\vartheta)}R_c C e^{A_c(\vartheta-t_{j+1})}d\vartheta\!\right)\!e^{-A_c(t-t_{j+1})}\!,
\end{array}
\end{equation}
which is well-defined for \red{all $(t, j) \in \mathcal{S}_c$}. Then existence of the solution depends on that of the boundary condition $\Upsilon(t_{j+1}, j)$, or equivalently, existence of $\Upsilon_{j+1}$ as $\Upsilon(t_{j+1}, j) = Q_d^{-1}\Upsilon_{j+1} A_d - Q_d^{-1}R_d C$ with $Q_d$ invertible by assumption. Now let $t = t_j$, substituting $\Upsilon(t_{j+1}, j)$ in the right-hand side of~\eqref{equ:UpsilonExpInterval}, with $Q_d^{-1}\Upsilon_{j+1}A_d - Q_d^{-1}R_d C$, we obtain
\begin{equation}\label{equ:UpsilonIterat}
\begin{aligned}
\Upsilon_j =\,\,\,& e^{Q_c(t_j-t_{j+1})}Q_d^{-1}\Upsilon_{j+1}A_de^{A_c(t_{j+1}-t_j)} \\
&- e^{Q_c(t_j - t_{j+1})}Q_d^{-1}R_d C e^{A_c(t_{j+1} - t_j)} \\
&+\int_{t_j}^{t_{j+1}} \!\!e^{Q_c(t_{j}-\vartheta)}R_c Ce^{A_c(\vartheta-t_{j})}d\vartheta.
\end{aligned}
\end{equation}
Therefore, we can find a solution $\Upsilon$ to~\eqref{equ:HybUpsilon} if $\Upsilon_j$ exists for any $j \in \mathbb{Z}$. To this end, we iterate and replace~\eqref{equ:UpsilonIterat} for $\Upsilon_{j+1}, \Upsilon_{j+2}, \cdots$, resulting in the relation
\begin{equation}\label{equ:UpsilonkExp}
\begin{array}{l}
\Upsilon_{j} \!=\! \prod_{\overrightarrow{i=j}}^{+\infty} \left(e^{-Q_c(t_{i+1}-t_{i})}Q_d^{-1}\right)\! \Upsilon_{+\infty}\prod_{\overleftarrow{i=j}}^{+\infty}\left(\!A_de^{A_c(t_{i+1}-t_{i})}\right)\\
\qquad +\, K_{\Upsilon\!,j},
\end{array}
\end{equation}
with $K_{\hat{\Upsilon}\!, j}$ as in~\eqref{equ:UpsilonkKupsilon} and $\Upsilon_{+\infty} \in \mathbb{R}^{\nu \times n}$ an arbitrary value when $j \to +\infty$. Then similarly to~\eqref{equ:PikExp}, as Assumption~\ref{asmp:filtQR} holds and the zero equilibrium of system~\eqref{equ:HybFOM} is exponentially stable, the first addend in~\eqref{equ:UpsilonkExp} is zero while the term $K_{\hat{\Upsilon}\!, j}$ is finite for all $j \in \mathbb{Z}$. Then by the definition in~\cite{ref:isidori2008steady}, $\hat{\Upsilon}$ satisfying~\eqref{equ:UpsilonBodCond} is the steady-state solution of~\eqref{equ:HybUpsilon} and is bounded for all times as $K_{\hat{\Upsilon}\!,j}$ is bounded whenever $j \to \pm \infty$. This solution, parameterised by~\eqref{equ:UpsilonBodCond}, is also unique because both the characterisation of the boundary condition after each jump, as in~\eqref{equ:UpsilonBodCond}, and the flow forward with time, as in~\eqref{equ:UpsilonExpInterval}, admit unique solutions under Assumption~\ref{asmp:jumpSpacing}. To show that this solution $\hat{\Upsilon}$ is the only bounded solution to~\eqref{equ:HybUpsilon}, we assume, by contradiction, that there exists another bounded solution $\Upsilon$ solving~\eqref{equ:HybUpsilon}. Then note that $\Xi:= \Upsilon - \hat{\Upsilon}$ satisfies
\begin{equation}\label{equ:HybUpsilonError}
    \begin{array}{lcll}
    \dot{\Xi} &=& Q_c \Xi - \Xi A_c, &\quad \blue{(t, j) \in \mathcal{S}_c}, \\
    \Xi^+  &=& Q_d \Xi A_d^{-1}, &\quad \blue{(t, j) \in \mathcal{S}_d}.
    \end{array}
\end{equation}
Similarly to~\eqref{equ:HybPiError}, the evolution of $\Xi$ can be explicitly characterised by $\Xi(t, j) \!=\!e^{Q_c(t-t_j)} \hat{\Xi}_j e^{-A_c(t-t_j)}$ where
\begin{equation*}
\begin{array}{l}
\hat{\Xi}_j :=\prod_{\overleftarrow{i=\underline{j}}}^{j} \left(Q_de^{Q_c(t_{i}-t_{i-1})}\right)\Xi_{\underline{j}}\prod_{\overrightarrow{i=\underline{j}}}^{j}\left(e^{-A_c(t_{i}-t_{i-1})}A_d^{-1}\right).
\end{array}
\end{equation*}
for any integer $\underline{j} \leq j$. Since Assumption~\ref{asmp:filtQR} holds and the zero equilibrium of system~\eqref{equ:HybFOM} is stable, the first product sequence is bounded by hypothesis, whereas the second product sequence is diverging with $t + j \to +\infty$. Consequently, $\lim_{t+j \to +\infty} \Xi(t, j)$ $= +\infty$. As the solution $\hat{\Upsilon}$ is bound, this result contradicts boundedness of $\Upsilon$, indicating that $\hat{\Upsilon}$ is the only bounded solution to~\eqref{equ:HybUpsilon}. Finally, to show $\hat{\Upsilon}$ is backward attractive, a similar iteration as in~\eqref{equ:UpsilonkExp} yields that
$\Xi(t, j) = e^{-Q_c(t_j-t)} \Xi(t_{j+1}, j) $ $ e^{A_c(t_j-t)} = e^{-Q_c(t_j-t)} Q_d^{-1} \hat{\Xi}_{j+1} A_d e^{A_c(t_j-t)}$ where
\begin{equation*}
\begin{array}{l}
\hat{\Xi}_{j+1} \!:=\!\prod_{\overrightarrow{i=j}}^{\bar{j}}\! \left(e^{-Q_c(t_{i+1}-t_{i})}Q_d^{-1}\right) \!\Xi_{\bar{j}}\prod_{\overleftarrow{i=j}}^{\bar{j}}\! \left(A_de^{A_c(t_{i+1}-t_{i})}\right)\!.
\end{array}
\end{equation*}
for any integer $\bar{j} \geq j$. This result implies that $\lim_{t+j \to -\infty} \Xi(t, j) = 0$ for any $\Xi_{\bar{j}} \in \mathbb{R}^{\nu \times n}$. Equivalently, for any final condition $\Upsilon(\bar{t}, \bar{j}) \in \mathbb{R}^{\nu \times n}$ with  $(\bar{t}, \bar{j}) \in \mathcal{H}$, we have $\lim_{t+j \to -\infty} \Upsilon(t, j) - \hat{\Upsilon}(t, j) = 0$, \textit{i.e.}, the bounded solution $\hat{\Upsilon}$ is backward attractive.
\end{proof}

\begin{remark}
    Similarly to $\hat{\Pi}$ discussed in Remark~\ref{rmk:SylvGeneralise}, the solution $\hat{\Upsilon}$ studied in Lemma~\ref{lem:UpsilonExist} is a hybrid generalisation of the solution of the Sylvester equation~\eqref{equ:SylvUpsilonLTI}. More specifically, if the system~\eqref{equ:HybFOM} and the filter~\eqref{equ:HybFilt} are purely continuous-time systems without jumps, the boundary condition~(\ref{equ:UpsilonkKupsilon}) with $j = 0$ reduces to
    $$
    \hat{\Upsilon}_{0}= K_{\hat{\Upsilon}, 0} = \int_0^{+\infty}e^{-Q_c\vartheta}R_c C e^{A_c\vartheta}d\vartheta,
    $$
    which is the solution to~\eqref{equ:SylvPiLTI} with $A_c = A$, $Q_c = Q$, and $R_c = R$.
    Moreover, if system~\eqref{equ:HybFOM} is continuous-time, \textit{i.e.}, $A_d = I_n$ and $B_d = 0$, and the filter~\eqref{equ:HybFilt} has $R_d = R_c$, then it is possible to reformulate the steady-state solution $\hat{\Upsilon}$ of~\eqref{equ:HybUpsilon} by~\eqref{equ:explicitUpsilon}, see~\cite{ref:mao2024model} for more detail. This generalised Sylvester equation, explicitly characterised by~\eqref{equ:explicitUpsilon}, has also been used for the stabilisation of cascade hybrid systems~\cite{ref:niu2025stabilisation}.
\end{remark}

The hybrid solution $\hat{\Upsilon}$ of~\eqref{equ:HybUpsilon} characterises the steady-state response of the swapped interconnection as follows.

\begin{theorem} 
\label{thm:UpsilonSS}
Consider the swapped interconnection of system~(\ref{equ:HybFOM}) and the hybrid filter~(\ref{equ:HybFilt}) shown in Fig.~\ref{fig:MMIntHyb}(b). Suppose Assumptions~\ref{asmp:jumpSpacing} and~\ref{asmp:filtQR} hold, and the zero equilibrium of system~(\ref{equ:HybFOM}) is exponentially stable with $A_d$ invertible. Then for any non-zero signals $u_c$ and $u_d$ exponentially decaying to zero, the steady-state response of the filter~\eqref{equ:HybFilt} in the interconnection, initialised with $x_0 = 0$ and $\varpi_0 = 0$, satisfies $\lim_{t+j \to +\infty} \|d(t, j) - \varpi_{ss}(t, j)\| = 0$, where $d(t, j) \in \mathbb{R}^{\nu}$ is the solution to
\begin{equation}\label{equ:Hybd}
\left\{\,\begin{array}{lll}
\dot{d} &= Q_c d + \hat{\Upsilon} B_c u_c, &\quad \blue{(t, j) \in \mathcal{S}_c},\\
d^+ &= Q_d d + \hat{\Upsilon}^+ B_d u_d, &\quad \blue{(t, j) \in \mathcal{S}_d},
\end{array} \right.
\end{equation}
with the initial condition $d_0 = 0$ and $\hat{\Upsilon}$ the unique steady-state solution to~\eqref{equ:HybUpsilon}.

\end{theorem}
\begin{proof}
    Since the solution $\hat{\Upsilon}$  to~\eqref{equ:HybUpsilon} is bounded and piecewise absolutely continuous, the solution $d$ to~\eqref{equ:Hybd} also exists and is bounded for all $(t, j) \in \mathcal{H}$ as Assumption~\ref{asmp:filtQR} holds and the input signals $u_c$ and $u_d$ exponentially converge to zero. Then to prove $d$ satisfies $\lim_{t+j \to +\infty} \|d(t, j) - \varpi_{ss}(t, j)\| = 0$, we first show that $d$ in~\eqref{equ:Hybd}, initialised with $d_0 = 0$, is such that $d = \varpi + \hat{\Upsilon} x$. To see this note that
    \begin{equation}
    \begin{array}{rl}
       \overset{\smash{\raisebox{-3pt}{$\cdot$}}}{\overbrace{\varpi + \hat{\Upsilon} \, x}} &= Q_c \varpi + R_c C x + \dot{\hat{\Upsilon}}x + \hat{\Upsilon} (A_cx + B_c u_c) \\
        &= Q_c\varpi +  (R_c C + \dot{\hat{\Upsilon}} + \hat{\Upsilon} A_c)x + \hat{\Upsilon} B_c u_c \\
        &= Q_c(\varpi + \hat{\Upsilon} x) + \hat{\Upsilon} B_c u_c
    \end{array}
    \end{equation}
    \red{whenever $(t, j) \in \mathcal{S}_c$}, and
    \begin{equation}
    \begin{array}{rl}
        (\varpi + \hat{\Upsilon} x)^+ &= Q_d \varpi + R_d C x + \hat{\Upsilon}^+ A_d x + \hat{\Upsilon}^+ B_d u_d \\
        &= Q_d\varpi + Q_d\hat{\Upsilon} x + \hat{\Upsilon}^+ B_d u_d \\
        &= Q_d(\varpi + \hat{\Upsilon} x) + \hat{\Upsilon}^+ B_d u_d
    \end{array}
    \end{equation}
    \red{whenever $(t, j) \in \mathcal{S}_d$}. Since the dynamics of $\varpi + \hat{\Upsilon} x$ coincides with the dynamics of $d$ in~\eqref{equ:Hybd} with the initial condition $d_0 = \varpi_0 + \hat{\Upsilon}_0 x_0 = 0$, then $d = \varpi + \hat{\Upsilon} x$ for all $(t, j) \in \mathcal{H}$. Meanwhile, since the zero equilibrium of system~\eqref{equ:HybFOM} is exponentially stable, with inputs $u_c$ and $u_d$ exponentially converging to zero, we have $\lim_{t+j \to +\infty} x(t, j) = 0$ and therefore 
    $$\lim_{t+j \to +\infty}\! \|d(t, j) - \varpi_{ss}(t, j)\| \!=\!\!\! \lim_{t+j \to +\infty} \|\hat{\Upsilon}(t, j)x(t, j)\| \!=\! 0$$ 
    as $\hat{\Upsilon}$ is bounded for all times. This concludes the proof.
\end{proof}

Theorem~\ref{thm:UpsilonSS} indicates that the steady-state response $\varpi$ of the filter in the swapped interconnection in Fig.~\ref{fig:MMIntHyb}(b) is characterised by the dynamics of $\hat{\Upsilon} B_c$ and $\hat{\Upsilon}^{+} B_d$. Unlike $\hat{\Pi}$ in Lemma~\ref{lem:PiExist}, the bounded solution $\hat{\Upsilon}$ of~\eqref{equ:HybUpsilon}, as shown in Lemma~\ref{lem:UpsilonExist}, is the only bounded solution when $t + j \to \infty$ (which is guaranteed by the assumption that $A_d$ in~\eqref{equ:HybFOM} is invertible) and is backward attractive. These two differences are expected, as the solution $\Upsilon$ of~\eqref{equ:HybUpsilon} can be seen as a dual version of the solution $\Pi$ of~\eqref{equ:HybPi}, see~\cite{ref:simpson2025steady}.

\begin{remark}\label{rmk:UpsilonCompute}
    From a computational point of view, the boundary condition~\eqref{equ:UpsilonBodCond} for characterising the unique bounded solution $\hat{\Upsilon}$ is hard to obtain \blue{as its computation requires the sum of infinite terms as in~\eqref{equ:UpsilonkKupsilon}. However, it is not necessary to explicitly compute~\eqref{equ:UpsilonkKupsilon} to derive the trajectories of $\hat{\Upsilon}$.} Due to the ``backward attractivity'' in~\eqref{equ:UpsilonBackAttract}, the bounded solution $\hat{\Upsilon}$ can be approximated if one solves~\eqref{equ:HybUpsilon} backwards, starting with a final time sufficiently large.
\end{remark}

In Lemma~\ref{lem:UpsilonExist} and Theorem~\ref{thm:UpsilonSS}, the invertibility of the matrix $A_d$ of system~\eqref{equ:HybFOM} is needed to guarantee the uniqueness of the bounded solution to~\eqref{equ:HybUpsilon}. This uniqueness is important, as we obtain a one-to-one relation between the steady-state response of the swapped interconnection and the dynamics $\hat{\Upsilon} B_c$ and $\hat{\Upsilon}^+ B_d$, inspiring the definition of moment for the swapped interconnection.
\begin{definition}\label{def:HybMomentSwapped} \textbf{\blue{(Moment of Hybrid System in Swapped Interconnection).}} Consider the swapped interconnection of system~(\ref{equ:HybFOM}) and the hybrid filter~(\ref{equ:HybFilt}). Suppose Assumptions~\ref{asmp:jumpSpacing} and~\ref{asmp:filtQR} hold, and the zero equilibrium of system~(\ref{equ:HybFOM}) is exponentially stable with $A_d$ invertible. Then the pair $(\hat{\Upsilon} B_c, \hat{\Upsilon}^+ B_d)$ is called the moment of system~\eqref{equ:HybFOM} at the tuple $(Q_c, R_c, Q_d, R_d, \mathcal{H})$, where $\hat{\Upsilon}$ is the unique bounded solution of~\eqref{equ:HybUpsilon}.
\end{definition}

The resulting moment matching problem is solved by the following theorem.

\begin{theorem}\label{thm:swapROM}
    Consider system~(\ref{equ:HybFOM}) and the filter~(\ref{equ:HybFilt}). Suppose Assumptions~\ref{asmp:jumpSpacing} and~\ref{asmp:filtQR} hold, and
    the zero equilibrium of system~(\ref{equ:HybFOM}) is exponentially stable with $A_d$ invertible. Then the system described by the equations
    \begin{equation}\label{equ:swapROM}
        \!\!\!\Sigma^{\prime}:\left\{\,\begin{array}{rcll}
        \dot{\xi}\,\,\, &=& (Q_c - R_c H) \xi + \hat{\Upsilon} B_c u_c, &\quad \blue{(t, j) \in \mathcal{S}_c},\\
        \xi^+ &=& (Q_d - R_d H) \xi + \hat{\Upsilon}^+ B_d u_d, &\quad \blue{(t, j) \in \mathcal{S}_c},\\
        \psi\,\,\, &=& H\xi,
        \end{array}\right.
    \end{equation}
    with $\hat{\Upsilon}$ the bounded solution of (\ref{equ:HybUpsilon}), is a \textit{model of system~(\ref{equ:HybFOM}) matching the moment at} the tuple $(Q_c, R_c, Q_d, R_d, \mathcal{H})$, for any $H$ such that the zero equilibrium of system~(\ref{equ:directROM}) is exponentially stable and the matrix $Q_d - R_d H$ is invertible. Moreover, system~(\ref{equ:swapROM}) is a \textit{reduced order model of system~(\ref{equ:HybFOM}) matching the moment at $(Q_c, R_c, Q_d, R_d, \mathcal{H})$} and solves Problem~\ref{prob:swappedMM} if $\nu<n$.
\end{theorem}
\begin{proof}
    Under the assumptions of the theorem and by Theorem~\ref{thm:UpsilonSS}, the swapped interconnection of system~(\ref{equ:swapROM}) and the filter~\eqref{equ:HybFilt} has a steady-state response $\varpi_{ss}$ satisfying $\lim_{t+j \to +\infty} \|d_r(t, j) - \varpi_{ss}(t, j)\| = 0,$ when initialised with $x_0 = 0$ and $\varpi_0 = 0$, where $d_r$ is the bounded solution of
    \begin{equation}\label{equ:HybdROM}
        \left\{\,\begin{array}{lll}
        \dot{d}_r &= Q_c d_r + \hat{Y} \hat{\Upsilon} B_c u_c, &\quad \blue{(t, j) \in \mathcal{S}_c}, \\
        d_r^+ &= Q_d d_r + \hat{Y}^+ \hat{\Upsilon}^+ B_d u_d, &\quad \blue{(t, j) \in \mathcal{S}_d},
        \end{array}\right.
    \end{equation}
    with initial condition $d_r(0, 0) = 0$. Note that as $Q_d - R_d H$ is invertible, by applying Lemma~\ref{lem:UpsilonExist} to system~\eqref{equ:swapROM}, the matrix-valued function $Y(t, j) = \hat{Y}(t, j) \in \mathbb{R}^{\nu \times \nu}$ is the unique bounded solution of
    \begin{equation}
    \label{equ:HybY}
    \begin{array}{rcll}
    \dot{Y} &=& Q_c Y - R_c H - Y (Q_c - R_c H),\\
    Y^+(Q_d - R_d H) &=& Q_d Y - R_d H,
    \end{array}
    \end{equation}
    \red{where and in what follows, we omit the jump conditions for brevity.} We see that this unique bounded solution is $\hat{Y}(t, j) = I_\nu$ for all $(t, j) \in \mathcal{H}$. Then in~\eqref{equ:HybdROM} we have $\hat{Y} \hat{\Upsilon} B_c = \hat{\Upsilon} B_c$ and $\hat{Y}^+ \hat{\Upsilon}^+ B_d  = \hat{\Upsilon}^+ B_d$. Consequently, the dynamics of $d_r$ in~\eqref{equ:HybdROM} is identical to that of $d$ in~\eqref{equ:Hybd}, and the steady-state response of the filter recovers that of the filter in the swapped interconnection with system~\eqref{equ:HybFOM}, \textit{i.e.}, Problem~\ref{prob:swappedMM} is solved.

\end{proof}

Again, matching the steady-state in  Theorem~\ref{thm:swapROM} is equivalent to the matching of the moment of system~\eqref{equ:HybFOM} by system~\eqref{equ:HybdROM}, with its moment $(\hat{Y} \hat{\Upsilon} B_c, \hat{Y}^+ \hat{\Upsilon}^+ B_d) = (\hat{\Upsilon} B_c, \hat{\Upsilon}^+ B_d)$. Note also that the vectors $R_c$ and $R_d$ are free parameters and can be used to achieve additional properties for the family of reduced order models. For example, they can be designed to achieve two-sided moment matching, as discussed in the next section.

\section{Two-Sided Moment Matching}\label{sec:2SidedMM}
On the basis of the direct and swapped interconnection moment matching results given by Theorems~\ref{thm:directROM} and~\ref{thm:swapROM}, we present the solution to the two-sided moment matching problem formulated in Problem~\ref{prob:twoSidedMM} relying on the following assumption.

\begin{assumption}\label{asmp:invertUpsPi}
    The function
    $\hat{\Upsilon}(t, j) \hat{\Pi}(t, j)$ is invertible for all $(t, j) \in \mathcal{H}$, where $\hat{\Pi}$ and $\hat{\Upsilon}$ are the unique bounded steady-state solutions of~\eqref{equ:HybPi} and~\eqref{equ:HybUpsilon}, respectively.
\end{assumption}

The invertibility of (constant) $\Upsilon \Pi$ is a standard assumption for the non-hybrid two-sided moment matching, see, \textit{e.g.},~\cite{Ionescu2016Two-SidedSystems} and~\cite{ref:scarciotti2024interconnection}. Under Assumption~\ref{asmp:invertUpsPi}, the following result holds.

\begin{theorem}\label{thm:2SidedROM}
    Consider system~(\ref{equ:HybFOM}), the signal generator~\eqref{equ:HybGen}, and the filter~(\ref{equ:HybFilt}). Suppose Assumptions~\ref{asmp:jumpSpacing},~\ref{asmp:genSJ},~\ref{asmp:filtQR}, and~\ref{asmp:invertUpsPi} hold and the zero equilibrium of system~(\ref{equ:HybFOM}) is exponentially stable with $A_d$ invertible. A \textit{model of system~(\ref{equ:HybFOM}) matching the moment at both $(S, L_c, J, L_d, \mathcal{H})$ and $(Q_c, R_c, Q_d, R_d, \mathcal{H})$} can be designed as
    \begin{itemize}
        \item[(i)] system~\eqref{equ:directROM} with $G_c = (\hat{\Upsilon} \hat{\Pi})^{-1} \hat{\Upsilon} B_c$ and $G_d = (\hat{\Upsilon}^+ \hat{\Pi}^+)^{-1} \hat{\Upsilon}^+ B_d$, if the matrix $J - G_d(t_j, j-1) L_d$ is invertible for all $j \in \mathbb{Z}$ and the zero equilibrium point is exponentially stable; 
        \item[(ii)] system~\eqref{equ:swapROM} with $H = C\hat{\Pi} (\hat{\Upsilon} \hat{\Pi})^{-1}$, if the matrix $Q_d - R_d H(t_j, j-1)$ is invertible for all $j \in \mathbb{Z}$ and the zero equilibrium point is exponentially stable; 
    \end{itemize}
    where $\hat{\Pi}$ and $\hat{\Upsilon}$ are the unique bounded steady-state solutions of~\eqref{equ:HybPi} and~\eqref{equ:HybUpsilon}, respectively. Moreover, A model of system~(\ref{equ:HybFOM}) at both $(S, L_c, J, L_d, \mathcal{H})$ and $(Q_c, R_c, Q_d, R_d, \mathcal{H})$ is a \textit{reduced order model of system~(\ref{equ:HybFOM}) matching the moment at the two tuples} and solves Problem~\ref{prob:twoSidedMM} if $\nu < n$.
\end{theorem}
\begin{proof}
    Since both designs rely on the term $\hat{\Upsilon} \hat{\Pi}$, we first study the dynamics of $\Phi := \hat{\Upsilon} \hat{\Pi}$. \red{Whenever $(t, j) \in \mathcal{S}_c$}, we have
    \begin{equation}\label{equ:PhiFlow}
    \begin{aligned}
        \dot{\Phi} &= (Q_c\hat{\Upsilon} - R_c C -\hat{\Upsilon} A_c) \hat{\Pi} + \hat{\Upsilon} (A_c\hat{\Pi} - \hat{\Pi} S + B_cL_c) \\
        &= Q_c \hat{\Upsilon} \hat{\Pi} - R_c C\hat{\Pi} -\hat{\Upsilon} \hat{\Pi} S + \hat{\Upsilon} B_c L_c \\
        &= Q_c \Phi - R_c C\hat{\Pi} - \Phi S + \hat{\Upsilon} B_c L_c.
    \end{aligned}
    \end{equation}
    \red{Whenever $(t, j) \in \mathcal{S}_d$}, by left-multiplying $\hat{\Upsilon}$ to the second equation of~\eqref{equ:HybPi}, and right-multiplying $\hat{\Pi}$ to the second equation of~\eqref{equ:HybUpsilon}, we obtain
    \begin{equation*}
    \begin{array}{rcll}
    \hat{\Upsilon}^+ \hat{\Pi}^+J &=& \hat{\Upsilon}^+A_d\hat{\Pi} + \hat{\Upsilon}^+B_dL_d, \\
    \hat{\Upsilon}^+A_d \hat{\Pi} &=& Q_d\hat{\Upsilon}\hat{\Pi} - R_d C\hat{\Pi}.
    \end{array}
    \end{equation*}
    Adding the two equations gives
    \begin{equation}\label{equ:PhiJump}
        \Phi^+J = Q_d \Phi + \hat{\Upsilon}^+ B_d L_d - R_d C \hat{\Pi}.
    \end{equation}
    Consider now case (i). By Theorem~\ref{thm:directROM}, system~\eqref{equ:directROM}, if $\nu < n$, is a reduced-order model of system~\eqref{equ:HybFOM} matching the moment at $(S, L_c, J, L_d, \mathcal{H})$. Then, to prove that it is also a reduced-order mode matching the moment at $(Q_c, R_c, Q_d, R_d, \mathcal{H})$, by Theorem~\ref{thm:swapROM} with $J - G_d(t_j, j-1) L_d$ invertible, we need to show that the unique bounded solution $Y(t, j) = \hat{Y}(t, j) \in \mathbb{R}^{\nu \times \nu}$ of
    \begin{equation}
    \label{equ:HybY2Side}
    \begin{array}{rcll}
    \dot{Y} &=& Q_c Y - R_c C \hat{\Pi} -Y (S - G_c L_c),\\
    Y^+(J - G_d L_d) &=& Q_d Y - R_d C \hat{\Pi},
    \end{array}
    \end{equation}
    also satisfies $\hat{Y} G_c = \hat{\Upsilon} B_c$ and $\hat{Y}^+ G_d = \hat{\Upsilon}^+ B_d$. By these two substitutions,~\eqref{equ:HybY2Side} can be simplified to
    \begin{equation*}
    \begin{array}{rcll}
    \dot{Y} &=& Q_c Y - R_c C \hat{\Pi} -YS + \Upsilon B_c L_c, \\
    Y^+J &=& Q_d Y + \Upsilon^+ B_d L_d - R_d C \hat{\Pi}.
    \end{array}
    \end{equation*}
    By comparing the last two equations with~\eqref{equ:PhiFlow} and~\eqref{equ:PhiJump}, the unique bounded solution of~\eqref{equ:HybY2Side} is $\hat{Y} = \Phi = \hat{\Upsilon} \hat{\Pi}$. Then, as $G_c = (\hat{\Upsilon} \hat{\Pi})^{-1} \hat{\Upsilon} B_c$ and $G_d = (\hat{\Upsilon}^+ \hat{\Pi}^+)^{-1} \hat{\Upsilon}^+ B_d$, we can see that $Y = \hat{\Upsilon} \hat{\Pi}$ also satisfies $Y G_c = \hat{\Upsilon} B_c$ and $Y^+ G_d = \hat{\hat{\Upsilon}}^+ B_d$. Case (ii) can be proved in a similar manner. By Theorems~\ref{thm:directROM} and~\ref{thm:swapROM}, system~\eqref{equ:HybFOM} with $H = C\hat{\Pi} (\hat{\Upsilon} \hat{\Pi})^{-1}$  and $Q_d - R_d H(t_j, j-1)$ invertible  solves Problem~\ref{prob:twoSidedMM} if the unique bounded steady-state solution $P(t, j) = \hat{P}(t, j) \in \mathbb{R}^{\nu \times \nu}$ of
    \begin{equation}
    \label{equ:HybP2Sided}
    \begin{array}{rcl}
    \dot P &=& (Q_c - R_c H) P - P S + \hat{\Upsilon} B_c L_c,\\
    P^+J &=& (Q_d - R_d H) P+ \hat{\Upsilon}^+ B_d L_d,
    \end{array}
    \end{equation}
    also satisfies $H \hat{P} = C \hat{\Pi}$. By substituting $H P$ in~\eqref{equ:HybP2Sided} with $C \hat{\Pi}$, we obtain
    \begin{equation*}
    \begin{array}{rcl}
    \dot P &=& Q_c P- R_c C \hat{\Pi} - P S + \hat{\Upsilon} B_c L_c,\\
    P^+J &=& Q_d P - R_d C \hat{\Pi} + \hat{\Upsilon}^+ B_d L_d,
    \end{array}
    \end{equation*}
    which, again, yields $\hat{P} = \hat{\Upsilon} \hat{\Pi}$. Thus, since $H = C\hat{\Pi} (\hat{\Upsilon} \hat{\Pi})^{-1}$, the equation $H\hat{P} = C\hat{\Pi}$ is also satisfied. This concludes the proof.
\end{proof}

\begin{remark}
    The invertibility of $J - (\hat{\Upsilon}_j \hat{\Pi}_j)^{-1} \hat{\Upsilon}_j B_d L_d$ and $Q_d - R_d C\hat{\Pi}(t_j, j-1) (\hat{\Upsilon}(t, j-1) \hat{\Pi}(t, j-1))^{-1}$, as indicated by Lemma~\ref{lem:UpsilonExist} and Theorem~\ref{thm:swapROM}, is needed to guarantee uniqueness of the solution of~\eqref{equ:HybY}. This requirement is not restrictive, as one can always adjust the values of $J$, $L_d$, $Q_d$, and $R_d$ to satisfy this invertibility condition. 
\end{remark}

\begin{remark}
    The two-sided moment matching result in Theorem~\ref{thm:2SidedROM} is a generalisation of the continuous-time results discussed at the end of Section~\ref{subsec:MMreview}. Following the direct and swapped cases in Theorems~\ref{thm:directROM} and~\ref{thm:swapROM}, the two-sided moment matching result in Theorem~\ref{thm:2SidedROM} also requires exponential stability of the resulting reduced-order model, which can be checked \textit{a posteriori} and depends on the dynamics of $\hat{\Pi}$ and $\hat{\Upsilon}$. In fact, unlike the one-sided moment matching results in Theorems~\ref{thm:directROM} and~\ref{thm:swapROM}, the proposed reduced-order model in Theorem~\ref{thm:2SidedROM} has no design freedom to guarantee its stability, as the freedom is sacrificed to achieve the two-sided matching. This is also the case for linear (continuous-time) systems. 
\end{remark}

As indicated by the proof of Theorem~\ref{thm:2SidedROM}, the system solving Problem~\ref{prob:twoSidedMM} matches the moment of system~\eqref{equ:HybFOM} for both direct and swapped interconnections, see Definitions~\ref{def:HybMomentDirect} and~\ref{def:HybMomentSwapped}. Note that all the results presented in this paper admit a natural extension to the time-varying case, \textit{i.e.}, the model reduction problems can be addressed in the same manner if systems~\eqref{equ:HybFOM},~\eqref{equ:HybGen}, and~\eqref{equ:HybFilt} are time-varying.

\section{The Special Case of Periodic Jumps}\label{sec:Periodic}

In this section, we specialise the proposed results to the case in which the hybrid time domain is periodic. In particular, we simplify the assumptions used in Lemmas~\ref{lem:PiExist} and~\ref{lem:UpsilonExist}, and reduce the complexity of computing $\hat{\Pi}$ and $\hat{\Upsilon}$ that solve~\eqref{equ:HybPi} and~\eqref{equ:HybUpsilon}, which are important for constructing the reduced-order model.

When the jumps of the hybrid system are periodic, \textit{i.e.}, $t_{j+1} - t_j = T > 0$, for all $j \in \mathbb{Z}$, the hybrid time domain in~\eqref{equ:hybTimeDomain} can be rewritten as
\begin{equation}\label{equ:HybTimeDomainPeriod}
   \mathcal{H}_{p}:=\{(t, j):t\in\left[jT,(j+1)T\right]\!,j\in\mathbb{Z}\}.
\end{equation}
Then the solution $\hat{\Pi}$ discussed in Lemma~\ref{lem:PiExist} boils down to a periodic solution as stated in what follows.

\begin{corollary}
\label{corol:PiExist}
Consider systems~(\ref{equ:HybFOM}) and~(\ref{equ:HybGen}) with the hybrid time domain~\eqref{equ:HybTimeDomainPeriod}. Suppose $\sigma\left(A_d e^{A_cT}\right)\subset\mathbb{D}_{<1}$ and system~\eqref{equ:HybGen} is neutrally stable, \textit{i.e.}, $\sigma(Je^{ST})\subset \mathbb{D}_{1}$. Then~\eqref{equ:HybPi} admits a unique $T$-periodic solution that coincides with its unique bounded steady-state solution $\hat{\Pi}$.
\end{corollary}

\begin{proof}
We first show the existence of a $T$-periodic solution. Since the hybrid time domain is $T$-periodic and the solution $\Pi$ within each interval $[jT, (j+1)T]$ can be parameterised by $\Pi_j$, as shown in~\eqref{equ:PiExpInterval}, $\Pi$ is periodic if and only if $\Pi_{j+1}=\Pi_j$ for all $j \in \mathbb{Z}$. Now rewriting equation~(\ref{equ:PiJump}) with the periodic hybrid time domain~\eqref{equ:HybTimeDomainPeriod} gives
\begin{equation}
\label{equ:PiPeriodic}
\begin{array}{l}
\!\!\Pi_{j}Je^{ST}\!=A_de^{A_c T}\Pi_{j}+B_dL_de^{ST}\\
\qquad\qquad\ +\, A_d\int_{0}^{T} e^{A_c(T-\vartheta)}B_cL_c e^{S\vartheta}d\vartheta,
\end{array}
\end{equation}
where $\Pi_{j+1}$ has been replaced by $\Pi_{j}$. As the sets of eigenvalues of $Je^{ST}$ and $A_de^{A_c T}$ are disjoint by hypothesis,~(\ref{equ:PiPeriodic}) is a Sylvester equation that has a unique solution, and therefore~\eqref{equ:HybPi} has a unique $T$-periodic solution, here denoted by $\bar{\Pi}$. To show that this $T$-periodic solution coincides with its unique steady-state solution $\hat{\Pi}$, by Lemma~\ref{lem:PiExist}, it suffices to show that $\bar{\Pi}_j = K_{\hat{\Pi},j}$ in~\eqref{equ:PiBodCond} for all $j \in \mathbb{Z}$, or equivalently, $K_{\hat{\Pi},j}$ is the unique solution to~\eqref{equ:PiPeriodic}. In fact, by denoting 
$L_\Pi = A_de^{A_cT}$, $R_\Pi = J e^{ST}$, and $$M_\Pi = B_d L_d J^{-1} + A_d (\int_{0}^{T} \!e^{A_c(T-\vartheta)}B_cL_c e^{S\vartheta}d\vartheta) J^{-1},$$ in the periodic case, $K_{\hat{\Pi},j}$ in~\eqref{equ:PiBodCond} boils down to
\begin{equation*}
\begin{array}{l}
    K_{\hat{\Pi},j} = \sum_{i=-\infty}^{j} L_\Pi^{j-i} M_\Pi R_\Pi^{i-j}.
\end{array}  
\end{equation*}
This result indicates that
\begin{equation*}
\begin{array}{l}
    K_{\hat{\Pi},j} - L_\Pi K_{\hat{\Pi},j} R_\Pi^{-1} \\
    =\! \sum_{i=-\infty}^{j}\! L_\Pi^{j-i}\! M_\Pi R_\Pi^{i-j} \!\!-\! \sum_{i=-\infty}^{j} \! L_\Pi^{j-i+1}\! M_\Pi R_\Pi^{i-j-1} \!\!=\! M_\Pi,
\end{array}  
\end{equation*}
implying that $K_{\hat{\Pi},j}$ is the unique solution to~\eqref{equ:PiPeriodic} and therefore the periodic solution satisfies $\bar{\Pi}_j = K_{\hat{\Pi},j}$. This concludes the proof.
\end{proof}

Similarly, the solution $\Upsilon$ discussed in Lemma~\ref{lem:UpsilonExist} can also be specialised to the periodic case.
\begin{corollary}\label{corol:UpsilonExist}
Consider systems~(\ref{equ:HybFOM}) and~(\ref{equ:HybFilt}) with the hybrid time domain~\eqref{equ:HybTimeDomainPeriod}. Suppose $\sigma\left(A_d e^{A_cT}\right)\subset\mathbb{D}_{<1}$ and system~\eqref{equ:HybFilt} is neutrally stable, \textit{i.e.}, $\sigma(Q_de^{Q_cT})\subset \mathbb{D}_{1}$. Then~\eqref{equ:HybUpsilon} admits a unique $T$-periodic solution that coincides with its unique bounded steady-state solution $\hat{\Upsilon}$.
\end{corollary}

\begin{proof}
The proof is similar to that of Corollary~\ref{corol:PiExist}. Since the hybrid time domain is $T$-periodic and the solution $\Upsilon$ within each interval $[jT, (j+1)T]$ can be parameterised by $\Upsilon_{j+1}$, as shown in~\eqref{equ:UpsilonExpInterval}, $\Upsilon$ is periodic if and only if $\Upsilon_{j+1}=\Upsilon_j$ for all $j \in \mathbb{Z}$. Rewriting equation~(\ref{equ:UpsilonIterat}) under the periodic hybrid time domain~\eqref{equ:HybTimeDomainPeriod} gives
\begin{equation}
\label{equ:UpsilonPeriodic}
\begin{array}{l}
    Q_d e^{Q_c T} \Upsilon_{j} = \Upsilon_{j}A_de^{A_c T} - R_d C e^{A_c T} \\
    \qquad\qquad\qquad\quad +\, Q_d\int_{0}^{T} e^{Q_c(T-\vartheta)}R_c Ce^{A_c\vartheta}d\vartheta,
    \end{array}
\end{equation}
where $\Upsilon_{j+1}$ has been replaced by $\Upsilon_{j}$. As the sets of eigenvalues of $Je^{ST}$ and $A_de^{A_c T}$ are disjoint by hypothesis, the Sylvester equation~(\ref{equ:UpsilonPeriodic}) has a unique solution, implying that~\eqref{equ:HybUpsilon} has a unique $T$-periodic (bounded) solution here denoted by $\bar{\Upsilon}$. Simiarly to the proof of Corollary~\ref{corol:PiExist}, we can prove that $K_{\hat{\Upsilon}\!, j}$ in \eqref{equ:UpsilonkKupsilon}, in the periodic case, is the unique solution to~\eqref{equ:UpsilonPeriodic}, and therefore $\bar{\Upsilon}_j = K_{\hat{\Upsilon}\!, j}$ for all $j \in \mathbb{Z}$. Then Lemma~\ref{lem:UpsilonExist} yields that $\bar{\Upsilon}_j$ coincides with the unique bounded steady-state solution $\hat{\Upsilon}$ of~\eqref{equ:HybUpsilon}.
\end{proof}

In the periodic case, Corollaries~\ref{corol:PiExist} and~\ref{corol:UpsilonExist} show that the boundary conditions $K_{\hat{\Pi}\!, j}$ and  $K_{\hat{\Upsilon}\!, j}$ that characterise the unique bounded steady-state solutions $\hat{\Pi}$ and $\hat{\Upsilon}$ of~\eqref{equ:HybPi} and~\eqref{equ:HybUpsilon} can be simply computed by solving the Sylvester equations~\eqref{equ:PiPeriodic} and~\eqref{equ:UpsilonPeriodic}, respectively. \blue{Consequently, computations of these solutions do not rely on attractivity-based approximation methods discussed in Remarks~\ref{rmk:PiCompute} and~\ref{rmk:UpsilonCompute}.}


\begin{remark}
Usually, when hybrid systems with periodic jumps are considered, the so-called conditions of non-resonance, \textit{i.e.}, $\sigma\left(J e^{ST}\right) \cap \sigma\left(A_d e^{A_cT}\right) = \emptyset$ and $\sigma\left(Q_d e^{Q_c T}\right) \cap \sigma\left(A_d e^{A_cT}\right) = \emptyset$, play a role in the definition of the steady-state. However, since the zero equilibria of signal generator~\eqref{equ:HybGen} and the filter~\eqref{equ:HybFilt}, in Corollaries~\ref{corol:PiExist} and~\ref{corol:UpsilonExist}, are assumed to be neutrally stable and the hybrid system~\eqref{equ:HybFOM} has an exponentially stable equilibrium at zero, the condition of non-resonance is automatically satisfied. In fact, we exclude signal generators and filters which produce decaying signals because in the problem of model reduction based on the steady-state notion of moment, a generator or filter with decaying dynamics does not contribute any steady-state information, see \cite{Ast:10}. 
\end{remark}

\section{A numerical example}
\label{sec:example}
In this section we present a numerical example to illustrate the results of the paper. \blue{All simulations were conducted in MATLAB using the variable-step solver \textit{ode45}.} In this example, we consider the two-sided moment matching problem for \blue{a MIMO} linear hybrid system~\eqref{equ:HybFOM} with system parameters randomly generated as
\begin{equation*}
\begin{aligned}
A_c&=\left[\begin{array}{rrrrrr}
   -1.280  &  0.953&    0.080&    0.646 &   1.207&    1.152\\
    0.124  & -1.635&   -0.007&    0.369 &   1.077&    1.529\\
   -0.037 &   0.179&   -0.891&   -0.046&    0.243&    0.758\\
   -1.146&   -0.634&   -0.591&   -0.490&    1.255&    0.706\\
   -0.796&   -1.232&   -0.533&   -1.315&   -0.980&   -0.754\\
   -1.467&   -0.889&   -0.172&   -1.359&    0.695&   -0.953
\end{array}\right],\\
A_d&=\left[\begin{array}{rrrrrr}
    0.141&   -0.488&   -0.317&   -0.137&    0.372&    0.100\\
   -0.242&    0.427&    0.106&    0.321&    0.492&  -0.256\\
   -0.515&    0.187&   -0.107&   -0.294&    0.254&    0.501\\
    0.289&    0.324&   -0.437&    0.004&    0.073&   0.222\\
    0.209&    0.262&    0.492&  -0.322&    0.336&   0.019\\
    0.234&   -0.182&    0.396&    0.350&    0.095&   0.724
\end{array}\right],\\
B_c&=\left[\begin{array}{rrrrrr}
    1.087 & 0.958 & 0.000 & -1.740 & -0.191 & 0.761 \\
    0.275 & -0.122 & -1.952 & -0.652 & -1.351 & 1.177
\end{array}\right]^{\top},\\
B_d&=\left[\begin{array}{rrrrrr}
    1.523 & -0.261 & 0.000 & 0.670 & -0.618 & -1.014 \\
    -0.755 & 0.114 & -1.337 & 0.408 & -0.948 & 0.616
\end{array}\right]^{\top},\\
C&=\left[\begin{array}{rrrrrr}
0.672 & 1.604 & 1.222 & 0.454 & -0.552 & -0.283 \\
0.757 & 0.993 & -0.198 & -1.665 & -1.084 & 1.653
\end{array}\right].
\end{aligned}
\end{equation*}
The signal generator~\eqref{equ:HybGen} and the filter~\eqref{equ:HybFilt} considered in this example have the parameter values as
\begin{equation}\label{equ:GenFilterParam}
    \begin{aligned}
        S &= 2, \quad\ \; J = 0.01, \quad L_c = L_d = \left[\begin{array}{cc}
        0.2  &  0.15 \\ 
        \end{array}\right], \\
        Q_c &= 1, \quad Q_d = 0.1, \quad \; \, R_c = R_d = \left[\begin{array}{cc}
        0.2  &  0.16 \\ 
        \end{array}\right]. 
    \end{aligned}
\end{equation}
Herein, we consider the state-dependent case. More specifically, all hybrid systems flow whenever $\omega(t, j) \in \mathcal{C}(t)$, where
\begin{equation}\label{equ:ct}
\begin{array}{l}
\!\!\!\mathcal{C}(t)\!=\! \left\{ s\in \mathbb{R}_{>0} \!:\! s\!<\!\sin\!\left(\!\frac{\sqrt{3}}{2} t\!\right) \!+\! 0.8 \nabla\!\left(\!\sqrt{5}(t-1)\right) \!+\! 3\right\}\!,
\end{array}   
\end{equation}
with $\nabla(t) := \frac{2}{\pi}\int_0^t \operatorname{sign}(\sin (\vartheta))d\vartheta-1$. Meanwhile, all systems jump whenever $\omega(t, j) \in \mathcal{D}(t)$ with $\mathcal{D}(t)$ defined as the boundary of $\mathcal{C}(t)$. \red{Note that $\mathcal{C}(t)$ in~\eqref{equ:ct} is not periodic as it is defined as the sum of a sinusoidal and a triangular wave with frequencies that are not rational multiples of each other}. As a consequence, the resulting jump instants $\{t_j\}_{j = -\infty}^{+\infty}$ triggered by $\omega$ are non-periodic. In this example, we randomly set the initial condition as $\omega_0 = 0.2739$. Fig.~\ref{fig:Omega_Bound} shows the time histories of the state of the generator $\omega$ (dashed blue) and of the jump set $\mathcal{D}(t)$ (solid green), \blue{implying that Assumption~\ref{asmp:jumpSpacing} holds, as the signal generator repeatedly enters the jump set $\mathcal{D}(t)$ and triggers the jump indefinitely. Moreover, the state $\omega$ never stays inside or consistently approaches to $\mathcal{D}(t)$ after each jump, \textit{i.e.}, dwell-time between any two successive jumps never vanishes to zero. Fig.~\ref{fig:Omega_Bound} also indicates the satisfaction of Assumption~\ref{asmp:genSJ}, which requires forward and backward boundedness of $\omega$.} Note also that in the swapped interconnection, the hybrid filter~\eqref{equ:HybFilt} with the parameter values in~\eqref{equ:GenFilterParam} satisfies Assumption~\ref{asmp:filtQR}. In the simulation, the decaying input is set to be $u = [e^{-t}, e^{-2t}]^\top$.


We now solve the two-sided problem, starting with the determination of the steady-state bounded solutions $\hat{\Pi}$ and $\hat{\Upsilon}$ of~\eqref{equ:HybPi} and~\eqref{equ:HybUpsilon}. As discussed in Remark~\ref{rmk:PiCompute}, although the initial condition is not important for $\hat{\Pi}$ due to its forward attractivity, we can still estimate~\eqref{equ:PiBodCond} by starting the simulation at $t = -10$ with a random initial condition $\Pi_i = [0.3377, 0.9001, 0.3692, 0.1112, 0.7803, 0.3897]^\top$. Consequently, the attractivity guarantees that the solution $\Pi$ converges to its steady state, obtaining a reasonable guess of $\hat{\Pi}_0$ in~\eqref{equ:PiBodCond} as 
\begin{equation*}
\hat{\Pi}_0 \!=\! \left[\begin{array}{rrrrrr}
    0.0480 & -0.0029 & -0.1166 & 0.0097 & 0.0028 & -0.1010
\end{array}\right]^\top\!.
\end{equation*}
Starting with this initial condition, Fig.~\ref{fig:Pi} shows the time histories of the six components of $\hat{\Pi}(t, j)$. It is easy to see that the components are not periodic. To illustrate the characterisation of the output response of system~(\ref{equ:HybFOM}) in the direct interconnection by~\eqref{equ:yss}, Fig.~\ref{fig:DirectOutSS} (top) shows the time histories of \red{the output $y = [y_1, y_2]^\top$} of system~(\ref{equ:HybFOM}) (solid blue) with randomly generated initial conditions 
\begin{equation}
\label{NOL-exx0}
x_0 \!=\! \left[\begin{array}{rrrrrr}
    4.3532 & 4.8615 & 0.5849 & 2.5416 & 3.7778 & 0.5499
\end{array}\right],
\end{equation}
and the time histories of the steady-state output \red{$y_{ss} = [y_{ss1}, y_{ss2}]^\top = C \hat{\Pi} \omega$} (dashed red). Meanwhile, Fig.~\ref{fig:DirectOutSS} (bottom) shows time history of the error $\|y - y_{ss}\|$ in semi-logarithmic scale. It can be seen that $y_{ss} = C \hat{\Pi} \omega$ characterises the steady-state of the output $y$ of the direct interconnection, where $C\hat{\Pi}$, by Definition~\ref{def:HybMomentDirect}, is the moment of system~\eqref{equ:HybFOM}.

\vspace{6mm}
\begin{figurehere}
\centering
\includegraphics[width=\columnwidth]{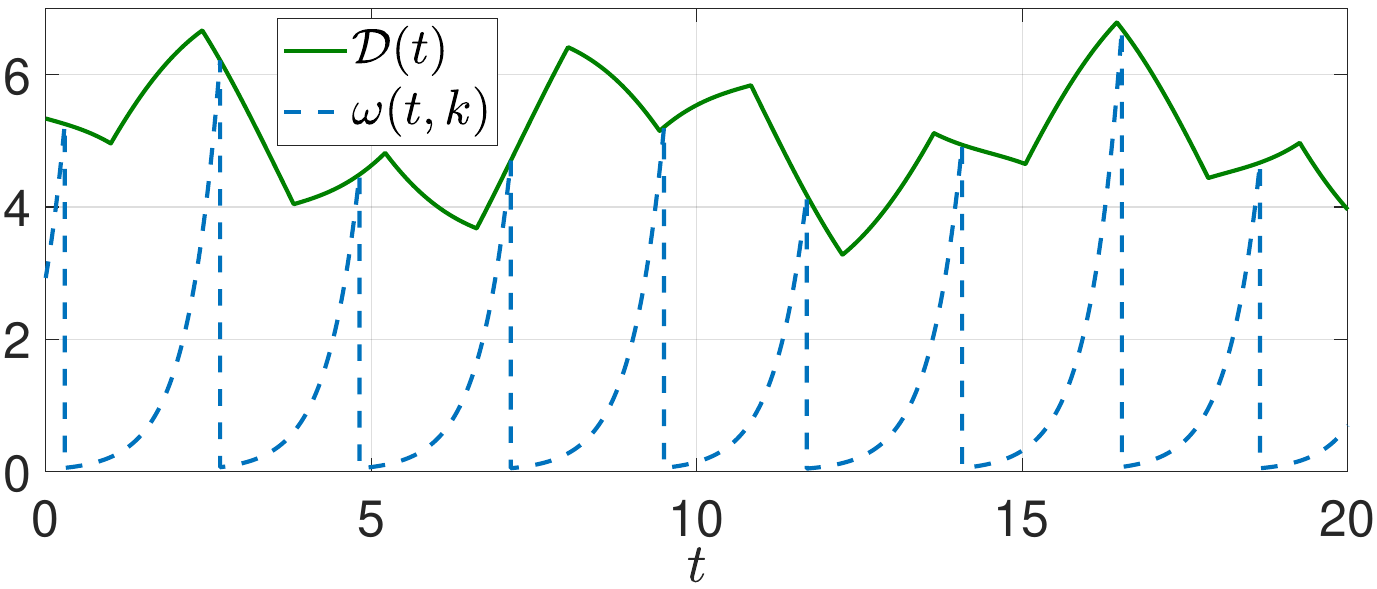}
\caption{Time histories of $\omega$ (dashed blue) and of the jump set $\mathcal{D}$ (solid green).}
\label{fig:Omega_Bound}
\end{figurehere}%

\begin{figurehere}
\centering
\includegraphics[width=\columnwidth]{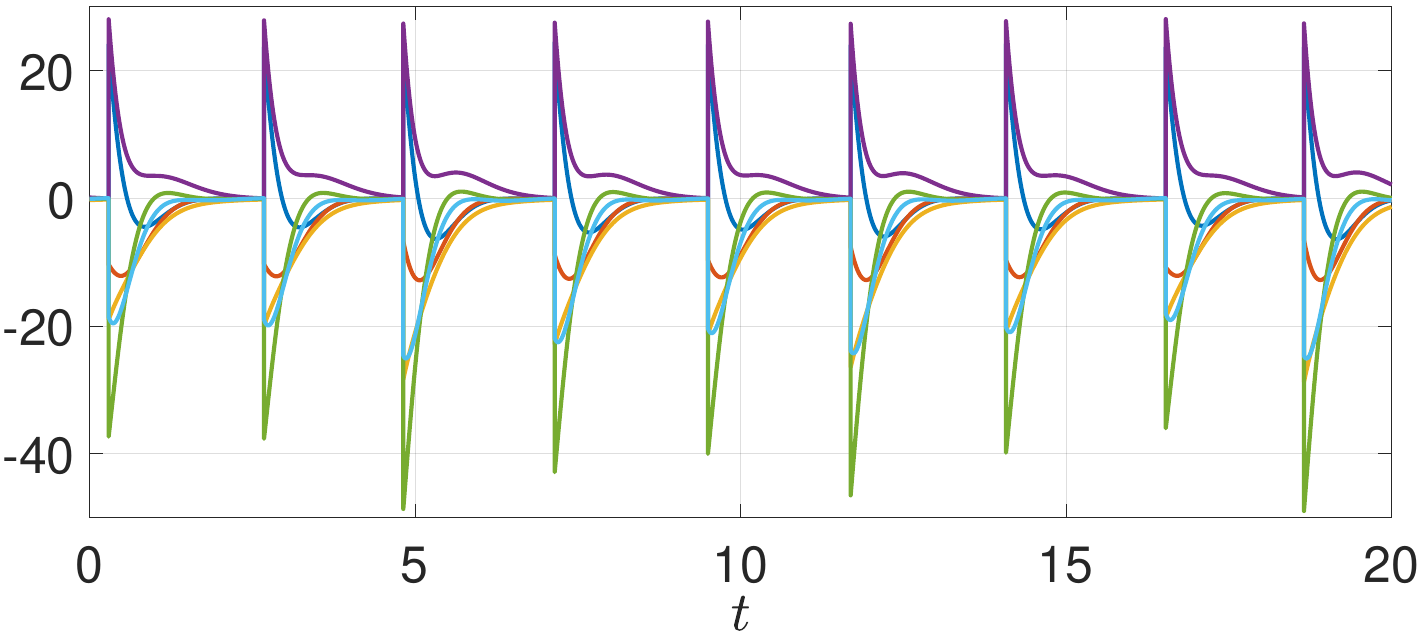}
\caption{Time histories of the six components of the matrix $\hat{\Pi}$.}
\label{fig:Pi}
\end{figurehere}%

\vspace{3mm}
\begin{figurehere}
\centering
\includegraphics[width=\columnwidth]{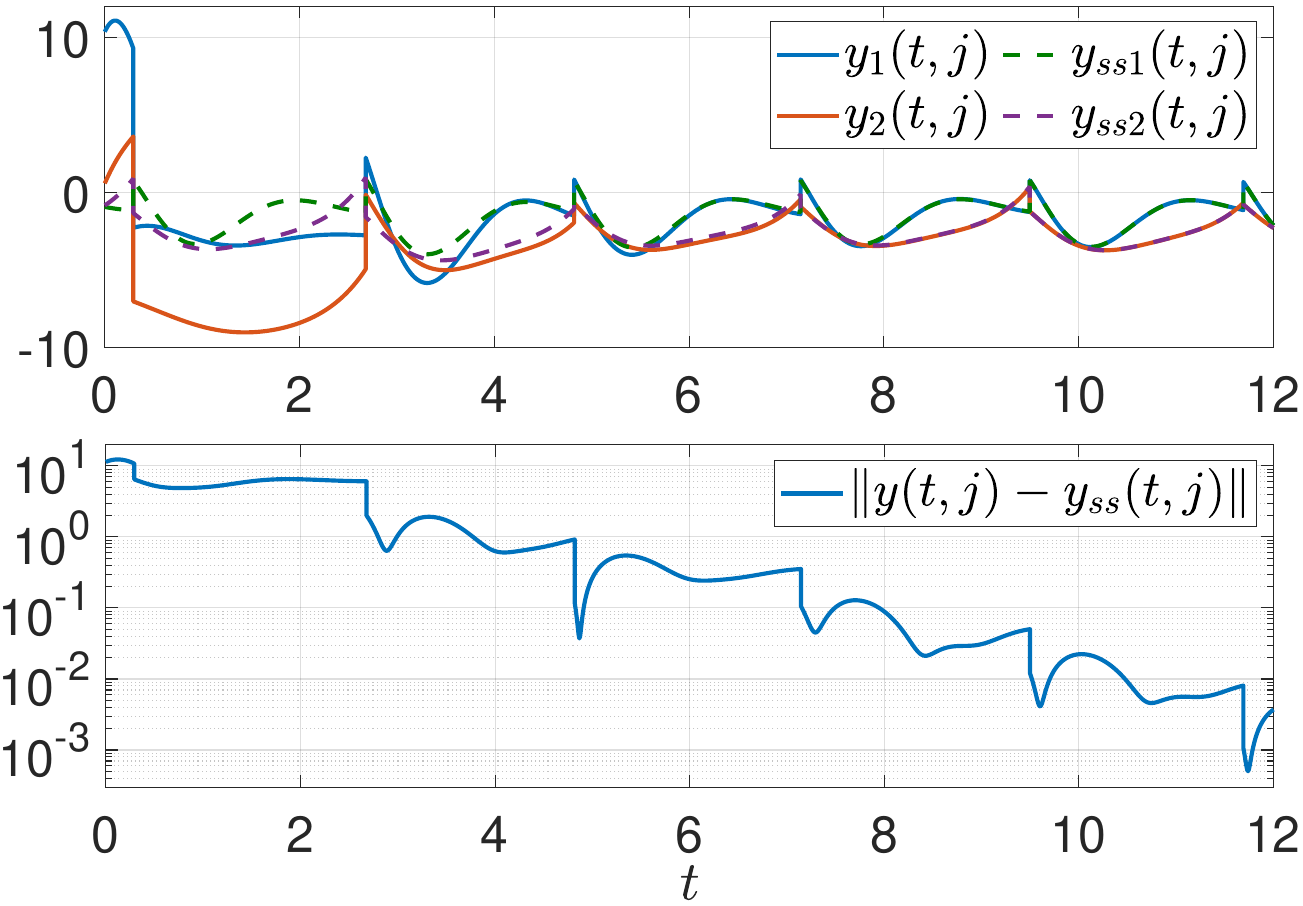}
\caption{Top graph: time histories of the output $y = [y_1, y_2]^\top$ of system~\eqref{equ:HybFOM} in the direct interconnection (solid blue) and time histories of the steady-state output $y_{ss} = [y_{ss1}, y_{ss2}]^\top = C \hat{\Pi} \omega$ (dashed red). Bottom graph: time history of the error $\|y - y_{ss}\|$ in semi-logarithmic scale.}
\label{fig:DirectOutSS}
\end{figurehere}%

We then determine the unique bounded solution to~\eqref{equ:HybUpsilon}. By Remark~\ref{rmk:UpsilonCompute}, the ``backward attractivity'' of $\hat{\Upsilon}$ allows computing its trajectories by solving~\eqref{equ:HybUpsilon} backwards in time, starting with a future time instant at a random final condition. Herein, we solve~\eqref{equ:HybUpsilon} backwards, starting at $t = 100$ with the random final condition $\Upsilon_{f} = [0.1504, -0.8804, -0.5304, -0.2937, 0.6424,$ $ -0.9692]$, resulting in the time history of the unique bounded solution $\hat{\Upsilon}$ as shown in Fig.~\ref{fig:Upsilon}. With this result, we then illustrate $\lim_{t+j \to +\infty} \|d(t, j) - \varpi_{ss}(t, j)\| = 0$ in Theorem~\ref{thm:UpsilonSS}. We numerically compute the state $d$ of system~\eqref{equ:Hybd} and the state $\varpi$ of the filter~\eqref{equ:HybFilt} with initial conditions $d_0 = \varpi_0 = 0$. Fig.~\ref{fig:SwappedOutSS} (top) displays the
time histories of the output $\varpi$ of the swapped interconnection (solid blue) and the state $d$ of system~\eqref{equ:Hybd} (dashed red), and Fig.~\ref{fig:SwappedOutSS} (bottom) shows the time history of their error $\|d - \varpi\|$ in semi-logarithmic scale. The results indicate that the response of $d$, which is characterised by the moment $(\hat{\Upsilon} B_c, \hat{\Upsilon}^+ B_d)$, characterises the steady-state response of the state $\varpi$ of the filter in the swapped connection.

\begin{figurehere}
\centering
\includegraphics[width=\columnwidth]{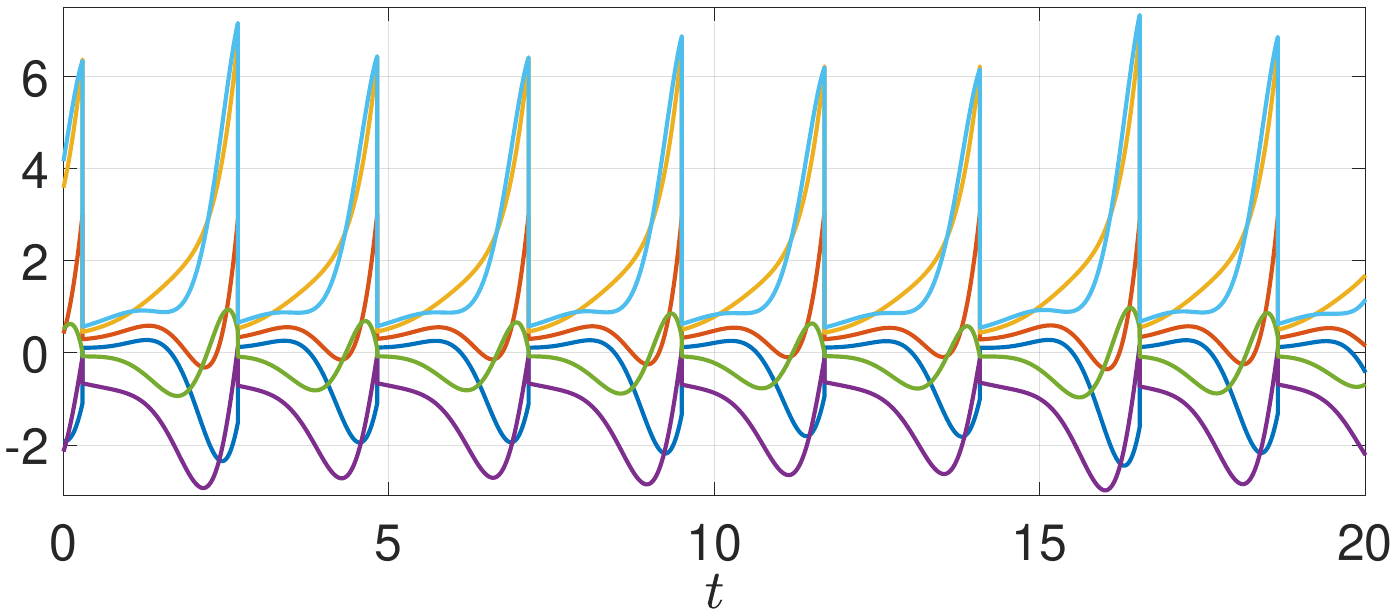}
\caption{Time histories of the six components of the matrix $\Upsilon$.}
\label{fig:Upsilon}
\end{figurehere}%

\begin{figurehere}
\centering
\includegraphics[width=\columnwidth]{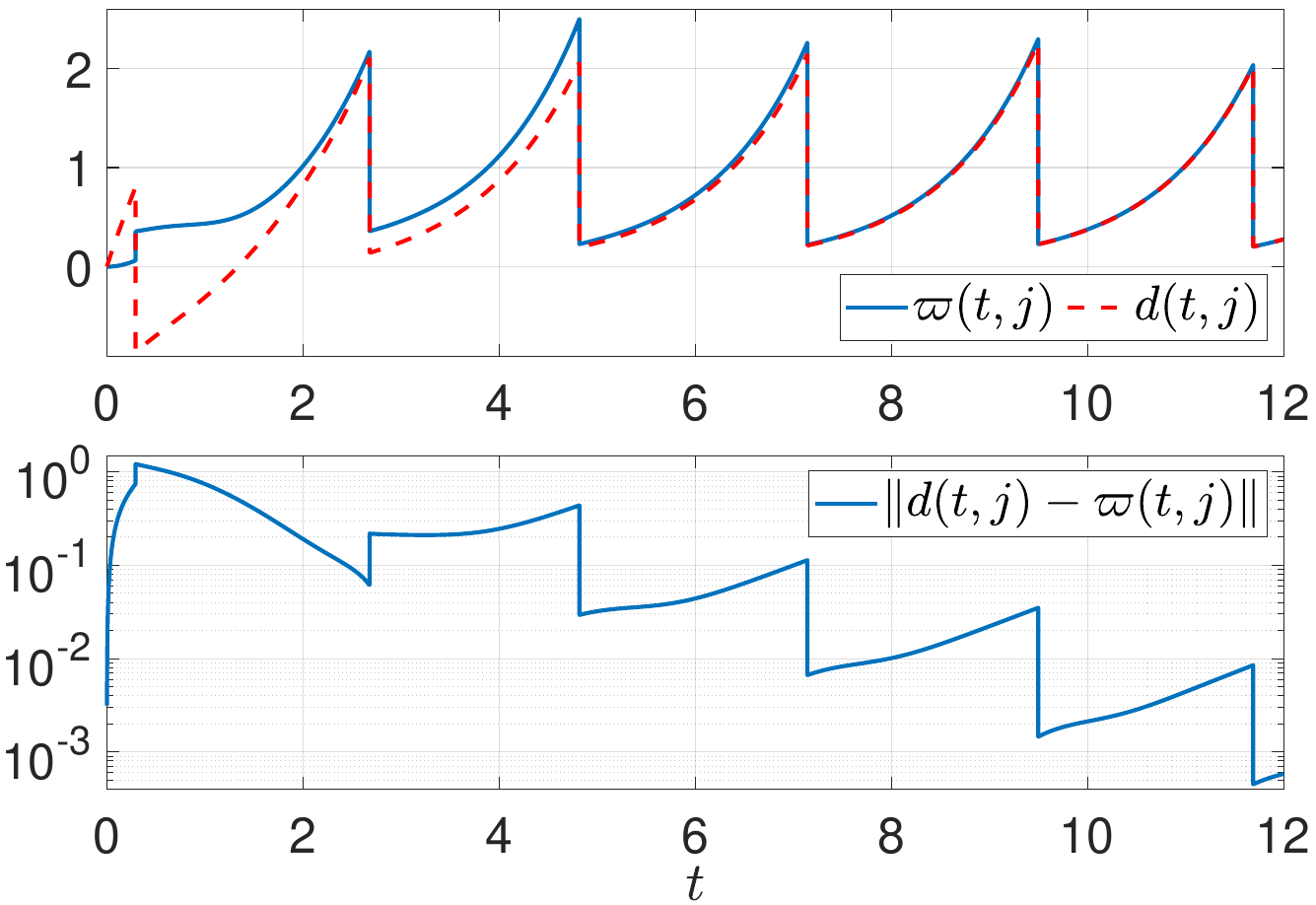}
\caption{Top graph: time histories of the state $\varpi$ of system~\eqref{equ:HybFilt} in the swapped interconnection (solid blue) and the state $d$ of system~\eqref{equ:Hybd} (dashed red). Bottom graph: time history of the error $\|d - \varpi\|$ in semi-logarithmic scale.}
\label{fig:SwappedOutSS}
\end{figurehere}%

\blue{With $\hat{\Pi}$ and $\hat{\Upsilon}$ numerically derived, one can construct reduced-order models achieving the one-sided matching problem in either Problem~\ref{prob:directMM} or~\ref{prob:swappedMM} by following Theorem~\ref{thm:directROM} or~\ref{thm:swapROM}, respectively. In this example, by following Theorem~\ref{thm:2SidedROM}, we show that it is possible to construct a single reduced-order model that solves both matching problems, \textit{i.e.}, the two-sided matching problem formulated in Problem~\ref{prob:twoSidedMM}.} To this end, we follow the design methods (i) and (ii) of Theorem~\ref{thm:2SidedROM} respectively, obtaining two reduced-order models (of order $\nu = 1$). \blue{Note that these two models are exponentially stable, which can be checked by simulating their free responses to see the existence of an exponentially-decaying bound on the state trajectory $\xi$ initialised with arbitrary initial conditions. As a result of the reduced-order model designed by method (i)}, Fig.~\ref{fig:ROM1Result} (top) shows the output responses of the direct (top) and swapped (bottom) interconnections of both system~\eqref{equ:HybFOM} (solid blue) and the reduced-order model (dashed red), where $\varpi_f$ and $\varpi_r$ denote the responses of the filter in the swapped interconnection with system~\eqref{equ:HybFOM} and the reduced-order model, respectively. This result indicates that the responses of the interconnections with system~\eqref{equ:HybFOM} are recovered by the reduced-order model. The same matching results, shown by Fig.~\ref{fig:ROM2Result}, exist in the reduced-order model derived by the design method (ii) of Theorem~\ref{thm:2SidedROM}. Similar to Fig.~\ref{fig:ROM1Result}, Fig.~\ref{fig:ROM2Result} depicts the output responses of the direct (top) and swapped (bottom) interconnections of both system~\eqref{equ:HybFOM} (solid blue) and the reduced-order model (dashed red), showing that the reduced-order model constructed by method (ii) solves the two-sided moment matching formulated in Problem~\ref{prob:twoSidedMM}. 

\vspace{2.5mm}
\begin{figurehere}
\centering
\includegraphics[width=\columnwidth]{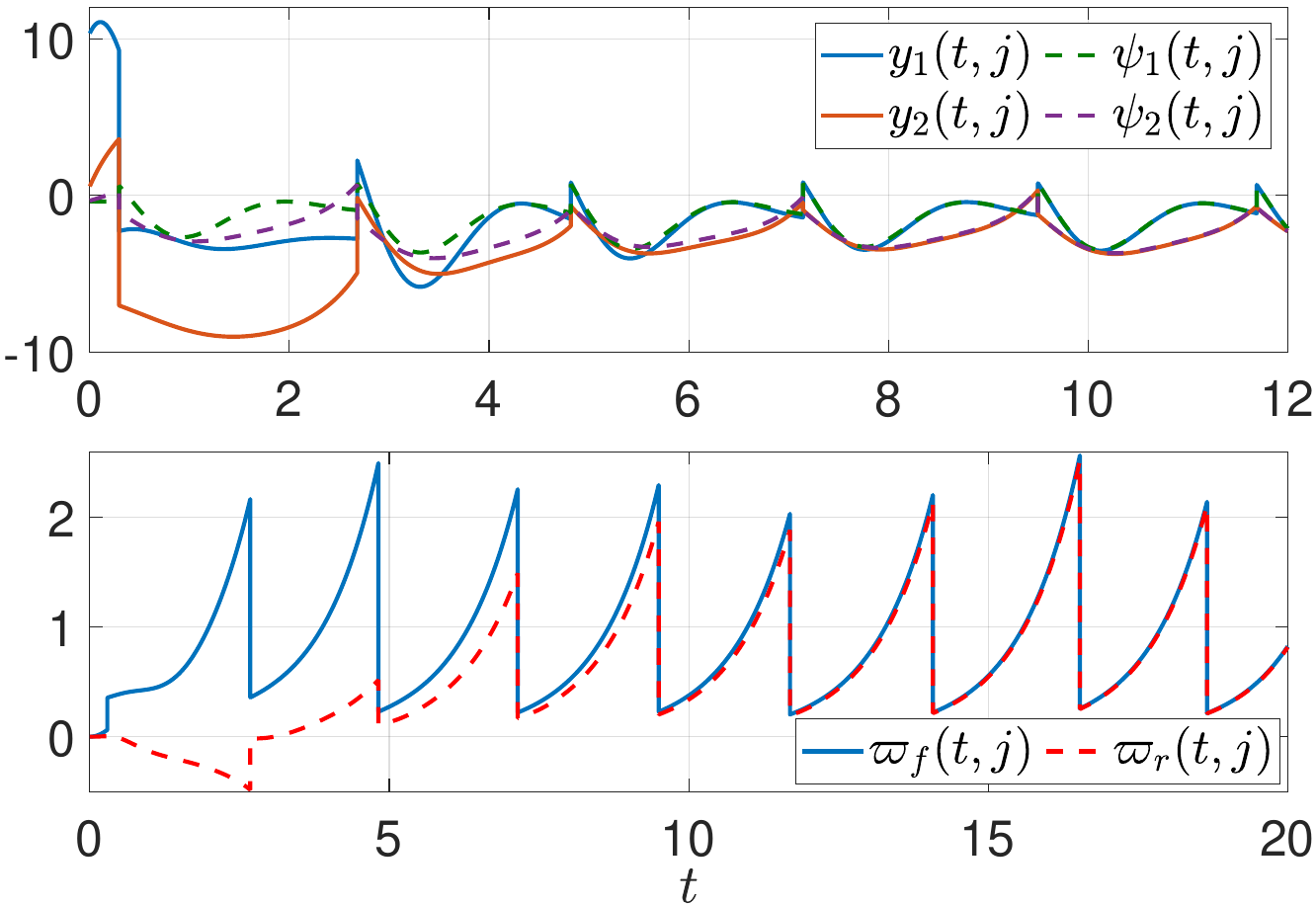}
\caption{Top graph: time histories of the output $y = [y_1, y_2]^\top$ of system~\eqref{equ:HybFOM} (solid blue) and the output $\psi = [\psi_1, \psi_2]^\top$ of the reduced-order model designed by method (i) of Theorem~\ref{thm:2SidedROM} (dashed red) in the direct interconnection. Bottom graph: time histories of the state of the filter when swapped interconnected with system~\eqref{equ:HybFOM} (the state $\varpi_f$ is in solid blue) and the reduced-order model designed by method (i) of Theorem~\ref{thm:2SidedROM} (the state $\varpi_r$ is in dashed red).}
\label{fig:ROM1Result}
\end{figurehere}%

\vspace{-2.5mm}
\begin{figurehere}
\centering
\includegraphics[width=\columnwidth]{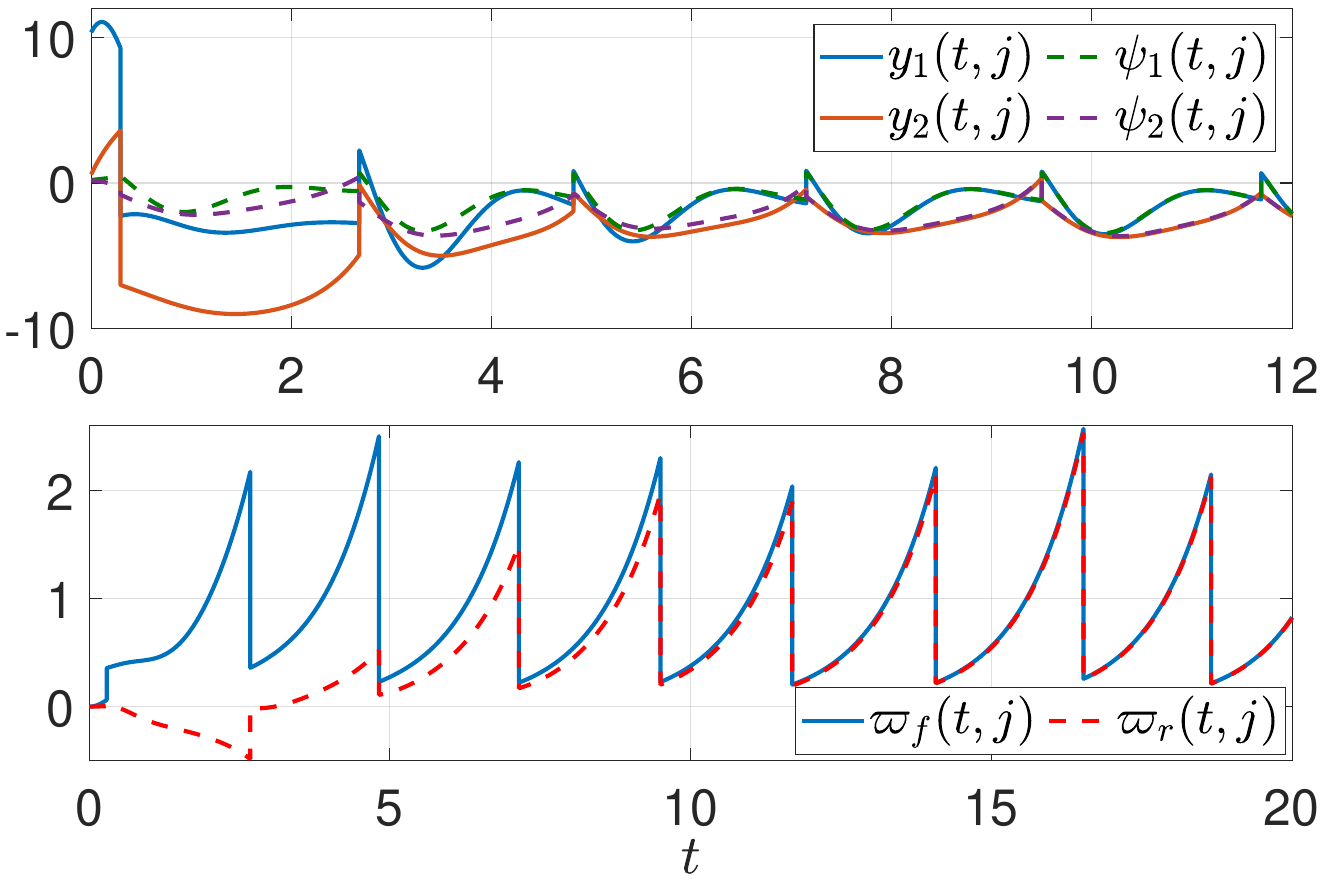}
\caption{Top graph: time histories of the output $y = [y_1, y_2]^\top$ of system~\eqref{equ:HybFOM} (solid blue) and the output $\psi = [\psi_1, \psi_2]^\top$ of the reduced-order model designed by method (ii) of Theorem~\ref{thm:2SidedROM} (dashed red) in the direct interconnection. Bottom graph: time histories of the state of the filter when swapped interconnected with system~\eqref{equ:HybFOM} (the state $\varpi_f$ is in solid blue) and the reduced-order model designed by method (ii) of Theorem~\ref{thm:2SidedROM} (the state $\varpi_r$ is in dashed red).}
\label{fig:ROM2Result}
\end{figurehere}%

\section{Conclusion}
\label{sec:concl}

In this paper we have addressed the model reduction problem for hybrid systems by means of moment matching. More specifically, we have first studied the moment matching problem for direct and swapped interconnections, deriving hybrid characterisations of steady-state responses of the interconnections and proposing families of hybrid reduced-order models. Based on these results, we have then solved the two-sided moment matching problem, showing two possible designs of the reduced-order model. In addition, all the solutions have been finally specialised to the case of hybrid systems with periodic jumps. A numerical simulation, in which the jumps of the hybrid systems are state-dependent, illustrates the results of the paper. Note that all proposed results can be naturally extended to the time-varying case.  

\bibliographystyle{ws-us}
\bibliography{bibdb}

\noindent\includegraphics[width=1in]{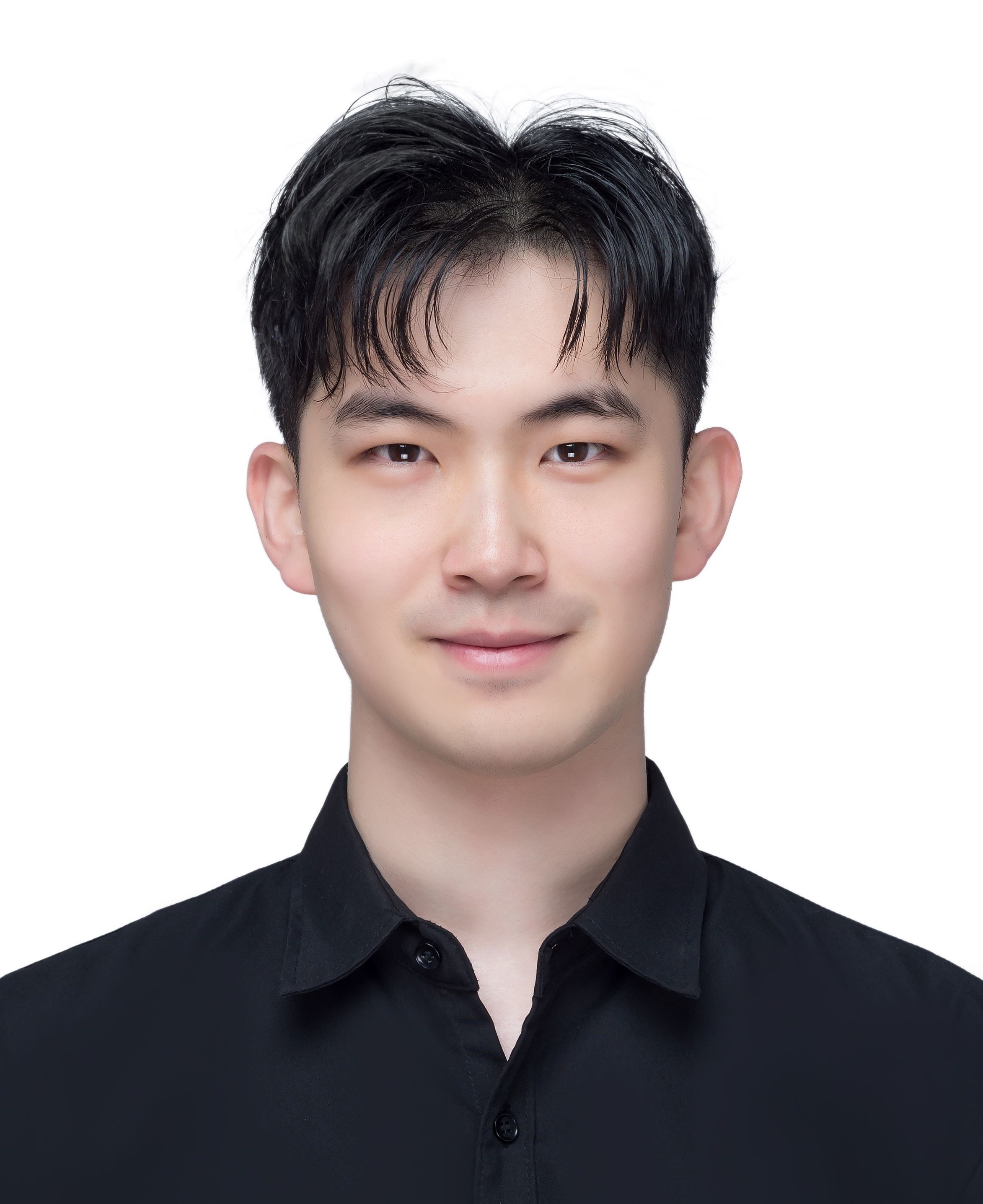}
{\bf Zirui Niu} (Graduate Student Member, IEEE) was born in Shandong, China, in 1997. He received the B.Eng. (Hons) degree in electrical and electronic engineering from the University of Liverpool, UK, in 2020, and the M.Sc. degree in control systems from Imperial College London, UK, in 2021. Since 2022, he has been pursuing a Ph.D. in control theory at Imperial College London, UK, supported by the CSC-Imperial Scholarship. His research interests include output regulation, hybrid systems, model reduction, approximate simulation of complex systems, and data-driven control. He was the recipient of the MSc Control Systems Outstanding Achievement Prize (2021) and the Hertha Ayrton Centenary Prize (2021) for outstanding performance in MSc Control Systems and the outstanding master’s thesis in the Electrical and Electronic Engineering Department at Imperial College London, respectively. \\

\noindent\includegraphics[width=1in]{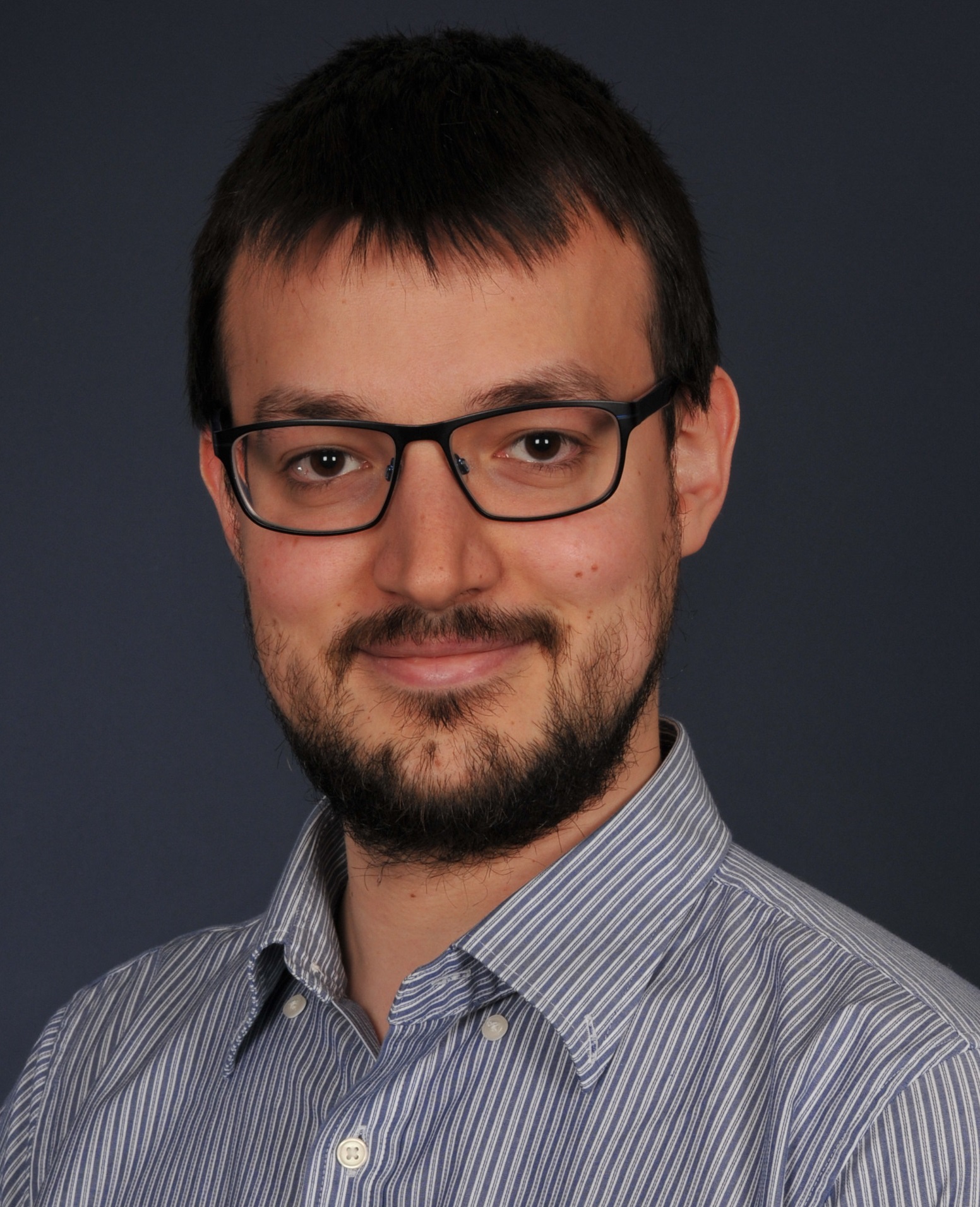}
{\bf Giordano Scarciotti} (Senior Member, IEEE) received his B.Sc. and M.Sc. degrees in Automation Engineering from the University of Rome ``Tor Vergata'', Italy, in 2010 and 2012, respectively. In 2012 he joined the Control and Power Group, Imperial College London, UK, where he obtained a Ph.D. degree in 2016. He also received an M.Sc. in Applied Mathematics from Imperial in 2020. He is currently an Associate Professor in Control Theory at Imperial. He was a visiting scholar at New York University in 2015, at University of California Santa Barbara in 2016, and a Visiting Fellow of Shanghai University in 2021-2022. He is the recipient of an Imperial College Junior Research Fellowship (2016), of the IET Control \& Automation PhD Award (2016), the Eryl Cadwaladr Davies Prize (2017), an ItalyMadeMe award (2017) and the IEEE Transactions on Control Systems Technology Outstanding Paper Award (2023). He is a member of the EUCA Conference Editorial Board, of the IFAC and IEEE CSS Technical Committees on Nonlinear Control Systems and has served in the International Programme Committees of multiple conferences. He is Associate Editor of Automatica and Guest Associate Editor of Nonlinear Analysis: Hybrid Systems. He was the National Organising Committee Chair for the EUCA European Control Conference (ECC) 2022, and of the 7th IFAC Conference on Analysis and Control of Nonlinear Dynamics and Chaos 2024, and the Invited Session Chair and Editor for the IFAC Symposium on Nonlinear Control Systems 2022 and 2025, respectively. He is the General Co-Chair of ECC 2029. \\

\noindent\includegraphics[width=1in]{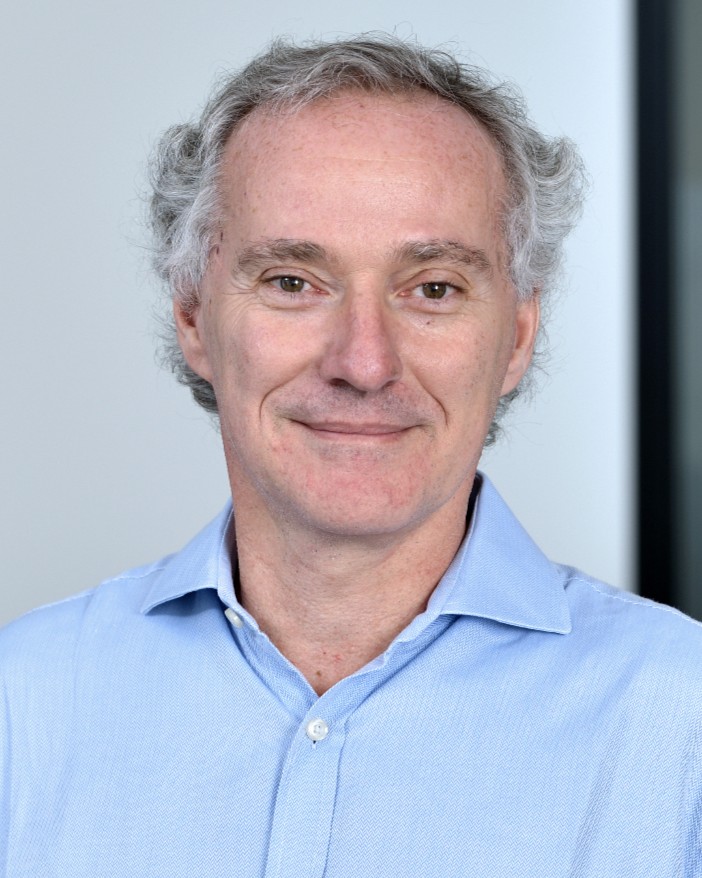}
{\bf Alessandro Astolfi} (Fellow, IEEE) graduated in Electronic engineering from the University of Rome in 1991. In 1992 he joined ETH-Zurich where he obtained a M.Sc. in Information Theory in 1995 and the Ph.D. degree with Medal of Honor in 1995. In 1996 he was awarded a Ph.D. from the University of Rome "La Sapienza". Since 1996 he has been with the Dept. of Electrical and Electronic Engineering of Imperial College London, where he is currently Professor of Nonlinear Control Theory. From 2010 to 2022 he served as Head of the Control and Power Group at Imperial College London and from 1998 to 2003 he was an Associate Professor at the Dept. of Electronics and Information of the Politecnico of Milano. Since 2005 he has also been a Professor at Dipartimento di Ingegneria Civile e Ingegneria Informatica, University of Rome Tor Vergata. His research focuses on mathematical control theory and applications, with particular emphasis on discontinuous stabilization, robust and adaptive control, observer design, and model reduction. He is the recipient of the IEEE CSS A. Ruberti Young Researcher Prize (2007), the IEEE RAS Googol Best New Application Paper Award (2009), the IEEE CSS George S. Axelby Outstanding Paper Award (2012), the Automatica Best Paper Award (2017), and the IEEE Transactions on Control Systems Technology Outstanding Paper Award (2023). He is a "Distinguished Member" of the IEEE CSS, IEEE Fellow, IFAC Fellow, IET Fellow, Member of the Academia Europaea, and of ITATEC. He served as Associate Editor for several journals; as Area Editor for the Int. J. of Adaptive Control and Signal Processing; as Senior Editor for the IEEE Trans. on Automatic Control; and as Editor-in-Chief for the European Journal of Control. He is currently Editor-in-Chief of the IEEE Trans. on Automatic Control (2018-). He served as Chair of the IEEE CSS Conference Editorial Board (2010-2017). He has served as Chair of the IEEE CSS Antonio Ruberti Young Researcher Prize (2015-2021); he is Vice Chair of the IFAC Technical Board (2020-2026).

\end{multicols}
\end{document}